\documentclass[reqno]{article}
\usepackage{amsthm}
\usepackage{amsmath}
\usepackage{amssymb}
\usepackage{graphicx}
\usepackage{color}
\usepackage{bbm}
\usepackage[top=2.5cm, bottom=2.5cm, left=2.5cm, right=2.5cm]{geometry}
\usepackage{enumerate}
\usepackage{faktor}
\usepackage{nicefrac}
\usepackage[nottoc]{tocbibind}
\usepackage{slashed}
\usepackage{hyperref}
\usepackage{marginnote}[2cm]

\newcommand{\N}{\mathbb{N}}
\newcommand{\R}{\mathbb{R}}

\newcommand{\Hp}{\mathcal{H}^+}

\newcommand{\vol}{\mathrm{vol}}

\newcommand{\volro}{\mathrm{vol}_{\{r = r_0\}}}
\newcommand{\nro}{n_{\{r=r_0\}}}

\newcommand{\nh}{n_{\mathcal{H}^+}}
\newcommand{\volh}{\mathrm{vol}_{\mathcal{H}^+}}

\begin{document}
	
	\numberwithin{equation}{section}
	\newtheorem{theorem}[equation]{Theorem}
	\newtheorem{remark}[equation]{Remark}
	\newtheorem{assumption}[equation]{Assumption}
	\newtheorem{claim}[equation]{Claim}
	\newtheorem{lemma}[equation]{Lemma}
	\newtheorem{definition}[equation]{Definition}
	\newtheorem{corollary}[equation]{Corollary}
	\newtheorem{proposition}[equation]{Proposition}
	\newtheorem*{theorem*}{Theorem}
	\newtheorem{conjecture}[equation]{Conjecture}

	\allowdisplaybreaks[1]

	\title{Generic blow-up results for the wave equation in the interior of a Schwarzschild black hole}
	\author{Grigorios Fournodavlos\thanks{Laboratoire Jacques-Louis Lions, Sorbonne Universit\'e,
			4 place Jussieu,
			Paris,
			75005,
			France}\; and Jan Sbierski\thanks{Mathematical Institute, 
			University of Oxford,
			Woodstock Road, 
			Oxford, 
			OX2 6GG,
			United Kingdom}}
	\date{\today}
	
	\maketitle
	
	\begin{abstract}
		We study the behaviour of smooth solutions to the wave equation, $\square_g\psi=0$, in the interior of a fixed Schwarzschild black hole. In particular, we obtain a full asymptotic expansion for all solutions towards $r=0$ and show that it is characterised by its first two leading terms, the principal logarithmic term and a bounded second order term. Moreover, we characterise an open set of initial data for which the corresponding solutions blow up logarithmically on the entirety of the singular hypersurface $\{r=0\}$. Our method is based on deriving weighted energy estimates in physical space and requires no symmetries of solutions. However, a key ingredient in our argument uses a precise analysis of the spherically symmetric part of the solution and a monotonicity property of spherically symmetric solutions in the interior. 
	\end{abstract}
	
	\tableofcontents

	\section{Introduction}
	
	The study of the linear wave equation on black hole backgrounds serves as a toy model for the study of gravitational perturbations of these backgrounds. This paper focuses on the study of waves in the interior of a Schwarzschild black hole. The Penrose diagram of a Schwarzschild black hole is given below in Figure \ref{FigSchw}; the interior corresponds to the top triangle.
	\begin{figure}[h]
		\centering
		\def\svgwidth{7cm}
		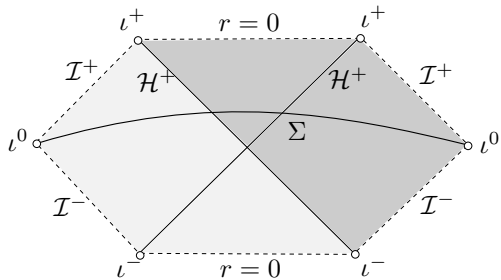
		\caption{The Penrose diagram of maximal analytic Schwarzschild} \label{FigSchw}
	\end{figure}
	
	The most widespread formulations of the influential black hole stability conjecture  and the strong cosmic censorship conjecture (for a neighbourhood of Schwarzschild) consider perturbations of the initial data for the Einstein equations given on a spacelike Cauchy hypersurface of the form of $\Sigma$ in Figure \ref{FigSchw}. This motivates the study of the \emph{Cauchy problem} for the wave equation on the Schwarzschild black hole \emph{with initial data given on $\Sigma$}. 
	
	The first boundedness and decay results for linear waves in the \emph{exterior} of the Schwarzschild black hole (i.e. the right diamond in Figure \ref{FigSchw}) were obtained in \cite{Wald79}, \cite{KayWald87}, and \cite{BlSo,BlSo2}, \cite{DafRod09a}. Many subsequent studies followed, see for example \cite{Luk10}, \cite{DoSchSo12}, \cite{Tat13}. Most relevant for this paper is the recent series \cite{AnArGa16a, AnArGa16b} which provides a very detailed analysis of linear waves in the exterior of the Schwarszchild black hole, in particular proving the upper and lower bounds of Price's law \cite{Price} along the event horizon.
	
	This paper complements the above study of waves in the exterior by presenting the first {rigorous} analysis of the behaviour of waves in the interior of the Schwarzschild black hole {(see however \cite{DoNo} for previous work in the physics literature)}. Our focus lies on the behaviour of the wave close to the curvature singularity at $\{r=0\}$ and we exhibit its generic blow-up profile there. More precisely, we derive the general asymptotic expansion to any finite order $N$ of all smooth solutions to the wave equation near the curvature singularity:
	\begin{align*}
	\psi = A(t,\omega)\log r + B(t,\omega) +\sum_{n=1}^N\big[\zeta_n(t,\omega)\log r+\eta_n(t,\omega)\big]r^n +\mathcal{O}(r^{N+1}|\log r|),\qquad \text{as $r\rightarrow0$},
	\end{align*}
	where $A,B,\zeta_n,\eta_n:\mathbb{R}\times\mathbb{S}^2\to\mathbb{R}$ are smooth functions, independent of $r$, and $\zeta_n,\eta_n$ are determined in terms $A,B$ and their derivatives, see Theorem \ref{ThmFullExp} for the precise formulas.\footnote{The remainder is also a smooth function, while its coordinate derivatives satisfy analogous bounds, see Theorem \ref{ThmFullExp} for more details.} In fact, we show that there exists a one to one correspondence between the solutions to the wave equation in the interior and the pair of functions $(A,B)$. 
	
	Assuming certain upper and lower bounds for the wave along the event horizons -- which according to \cite{AnArGa16a,AnArGa16b, AnArGa18} are satisfied by generic solutions  -- we show that $A$ is non-vanishing near the endpoints of the singular hypersurface $\{r=0\}$ and thus the wave blows up there logarithmically. 
	Moreover, we identify an open set of initial data which gives rise to solutions that blow up \emph{everywhere} at $\{r=0\}$. We obtain blow-up by exploiting a monotonicity property for the spherically symmetric part of the wave in the interior. To our knowledge, this monotonicity property first appeared in \cite{DafAnn}, where it was used to show blow-up along the Cauchy horizons of (dynamical) black hole interiors for the spherically symmetric Einstein-Maxwell-scalar field system. 
	
	
	The generic blow-up profile that we exhibit in the Schwarzschild interior
	shares some features with the behaviour of linear waves close to Big Bang singularities \cite{AFF,AllRen} and even with dynamical waves for near-FLRW \cite{RodSp} Big Bang singularities, where a similar logarithmic blow-up profile has been exhibited. 
	It is however in sharp contrast to the behaviour of linear waves observed in sub-extremal Reissner-Nordstr\"om and Kerr black hole interiors, where the local $\dot{H}^1$ energy blows up generically at the Cauchy horizons \cite{LukOh15,LukSb,DafShlap}, but the waves themselves extend continuously past the inner null boundaries of the black hole \cite{Fra14,Fra17,Hin15}. In the extremal case, axi-symmetric solutions to the wave equation have been shown to extend continuously up to and including the Cauchy horizons, having also finite local energy \cite{Gajic15a,Gajic15b}. Furthermore, smooth solutions to the wave equation in Reissner-Nordstr\"om-de Sitter and Kerr-de Sitter black hole interiors were shown \cite{HinVas15} to extend continuously to the Cauchy horizons and decay exponentially to a constant along them. 
	
{Although the linear wave equation is regarded as a `poor' linearisation of the Einstein equations, the behaviour of linear waves in black hole interiors can potentially give some insight concerning the much harder non-linear stability problem and extendibility questions related to the strong cosmic censorship conjecture. Indeed, this has been proven to be the case in the recent breakthrough work \cite{DafLuk}, the first in an announced series of three papers, proving the $C^0$-stability of the Kerr Cauchy horizons, under the assumption that the dynamical exterior asymptotes to a rotating Kerr black hole at timelike infinity. Note that this is in perfect analogy to the $C^0$-extendibility of linear waves across null boundaries that we mentioned above. 
The situation is however drastically different for the Schwarzschild spacetime: in contrast to the exact Kerr spacetime, the metric here does not extend continuously in the interior beyond the maximal globally hyperbolic development, \cite{Sbier}, and, moreover, generic linear waves become unbounded in the interior. This indicates a radically different nature of the non-linear stability problem for the Schwarzschild interior compared to the Kerr interior. 
It is claimed that, should spacelike singularities occur, they will not in general be of Schwarzschild type, but they will instead exhibit what is called mixmaster or oscillatory behaviour (BKL), see \cite{BKL,IsenMon,Ring} and the references therein. However, for subclasses of spacetimes enjoying certain symmetries and or polarized conditions, see \cite{IsenMon}, spacelike singularities exhibiting AVTD behaviour\footnote{The Schwarzschild singularity is trivially AVTD (asymptotically velocity term dominated). Loosely speaking, these singularities can be viewed as Schwarzschild type, in the sense that the metric components have similar blowing up/vanishing rates, which vary between different spatial points along the singular hypersurface, while satisfying the classical Kasner relations. See also \cite{Four} for the backwards construction of  Schwarzschild type singularities in non-symmetric spacetimes.} are expected to form instead. In this latter case, what is called asymptotic velocity \cite{KichRen}, can be seen to correspond to the coefficient $A(t,\omega)$ of the leading order logarithmic term in our expansion above, where the corresponding dynamical variable that has this logarithmic behaviour, to leading order at the singularity, satisfies a non-linear wave type equation.}

Finally, we note that it would be interesting to examine whether similar blow-up behaviour to the one established in the present paper  is exhibited for the spherically symmetric Einstein-scalar field system, where general black hole solutions obtained in \cite{Christ1,Christ2} contain a spacelike singularity in their black hole interior. {Indeed, the analytic and numerical results of \cite{Bu} seem to suggest this.}

	\subsection{The Schwarzschild black hole and notation} \label{SecNotation}
	
	We recall that the darker shaded region of the maximal analytic Schwarzschild spacetime $(M,g)$ depicted in Figure \ref{FigSchw} is given by the manifold $\R \times (0, \infty) \times \mathbb{S}^2$, with standard $(v,r, \theta, \varphi)$ coordinates, and metric
	\begin{equation} \label{Metricvr}
	g= -(1 - \frac{2m}{r})\, dv^2 + dv \otimes dr + dr \otimes dv + r^2 \, \big(d\theta ^2 + \sin^2\theta \, d\varphi^2\big) \;.
	\end{equation}
	Here, $m >0$ is the mass of the Schwarzschild black hole. These coordinates are also referred to as ingoing Eddington--Finkelstein coordinates. We denote the hypersurface $\{r = 2m\}$ in this coordinate patch by $\Hp$ and call it the (right future) event horizon. 
	
	The lighter shaded region of the maximal analytic Schwarzschild spacetime in Figure \ref{FigSchw} is in fact isometric to the darker shaded one; they are glued together across a null hypersurface. We refer the reader to \cite{HawkEllis} for the precise construction of the maximal analytic Schwarzschild spacetime $(M, g)$. Although we will state the results of this paper for waves defined on all of the maximal analytic Schwarzschild spacetime $(M,g)$, it suffices for the proof, by the symmetry of $(M,g)$ which interchanges the two asymptotically flat regions, to only consider the darker shaded region defined above. 
	
	Returning to the darker shaded region, the part $\{0 < r < 2m\}$ is called the interior of the Schwarzschild black hole and is the region where most of our analysis takes place. We introduce a function $r^*(r)$ in this region that satisfies $\frac{dr^*}{dr} = \frac{1}{1 - \frac{2m}{r}}$,  $\lim_{r \to 2m} r^*(r) = -\infty$ and $\lim_{r \to 0}r^*(r) =: r^*(0)$. Setting $t = v - r^*$, the metric \eqref{Metricvr} takes the well-known form
	\begin{equation*}
	g = -(1- \frac{2m}{r}) \, dt^2 + \frac{1}{1 - \frac{2m}{r}} \, dr^2 + r^2 \, \big(d\theta ^2 + \sin^2\theta \, d\varphi^2\big) \;.
	\end{equation*}
	{ The interior region $\{r<2m\}$ is naturally foliated by the level sets $\{r = r_0\}$, which are  spacelike hypersurfaces and here denoted by $\Sigma_{r_0}$.
	Moreover,} we will denote with $\Omega_k$, $k = 1,2,3$, a basis of normalised generators of the symmetry group of $\mathbb{S}^2$. They give rise to a set of Killing vector fields on $(M,g)$, for which we use the same notation.
	Let us also recall that the Killing vector field $\partial_t$ (with respect to the above $(t,r, \theta, \varphi)$  coordinates) extends regularly to all of $M$. We denote this extended vector field again by $\partial_t$. In particular it is tangent to the right future event horizon $\Hp$ and future directed null there. We will also need the vector field\footnote{The following notation $\frac{\partial}{\partial r}\Big|_v$ indicates that this is a partial derivative with respect to the $(v,r,\theta, \varphi)$ coordinate system.} $Y:= -\frac{\partial}{\partial r}\Big|_v$, which is future directed null and is transversal to the right event horizon $\Hp$ and satisfies $g(Y, \partial_v) = -1$.
	
	Moreover, we define the function  $u = r^* -t$ in the interior. In $(u,v)$-coordinates for the interior the vector field $Y$ reads
	\begin{equation}\label{DefY}
	Y= -\frac{\partial}{\partial r}\Big|_v = -\frac{2}{1 - \frac{2m}{r}} \frac{\partial}{\partial u}\Big|_v \;.
	\end{equation}
	For the later analysis it is also convenient to recall $\partial_v r = \partial_u r = 1 - \frac{2m}{r}$.
Choosing $(u,r,\theta, \varphi)$ coordinates in the interior one sees that this chart can be extended by analyticity to cover also the left asymptotically flat (diamond shaped) region in Figure \ref{FigSchw}. We denote the hypersurface $\{r=2m\}$ in this coordinate chart	also by $\Hp$ and call it the left future event horizon. However, if not specified, $\Hp$ will always refer to the right event horizon. The vector field $-\partial_t$ is tangent to the left event horizon and future directed there  and the vector field $\tilde{Y}:=-\frac{\partial}{\partial r}\Big|_u$ is transversal to the left event horizon and future directed.
	
	Let us also recall that the \emph{wave equation} is defined by
	\begin{equation} \label{WaveEquation}
	\Box_g \psi := \frac{1}{\sqrt{-\det g}}\partial_\mu \big(g^{\mu \nu} \sqrt{- \det g} \,\partial_\nu \psi\big) = 0 \;.
	\end{equation}
	
	We will often use the notation $f(x) \lesssim g(x)$ for two functions $f$ and $g$, by which we mean that there exists a constant $C>0$ such that $f(x) \leq C \cdot g(x)$ holds for all $x$. The notation $f(x) \gtrsim g(x)$ is defined analogously. Finally, the notation $f(x) \sim g(x)$ means that both $f(x) \lesssim g(x)$ and $f(x) \gtrsim g(x)$ hold. We will also use the big O notation $\mathcal{O}$ in this paper.

	\subsection{The main results and outline of the paper} \label{SecMainThm}
	
	Our focus lies on the behaviour of the wave close to the curvature singularity at $\{r=0\}$. In Section \ref{SecSphericalSymmetry} we begin with an analysis for spherically symmetric waves and exhibit a blow-up mechanism that exploits a monotonicity property for the spherically symmetric wave equation in the trapped region: if the null derivatives $\partial_u \psi $ and $\partial_v \psi$ have the same sign in the black hole interior, this sign is propagated to the future and the amplitude of these quantities is monotonically growing. We prove

	\begin{theorem}
		\label{ThmSphSym}
		
		Let $\psi$ be a smooth, spherically symmetric solution to the wave equation \eqref{WaveEquation} on the maximal analytic Schwarzschild spacetime.
		\begin{enumerate}[a)]
			\item Let $\Sigma'$ be a spherically symmetric achronal hypersurface  containing a portion $\Hp \cap \{v \geq v_0\}$, $v_0 \geq 1$, of the right event horizon and then transitioning through the interior to an analogous component of the left event horizon, see Figure \ref{FigSigmaPrime}. 
			Assume there exists a  $q > -1$ such that $\psi$ satisfies along $\Hp \cap \{v \geq v_0\}$
			\begin{equation*}
			\begin{aligned}
			\partial_t \psi &\gtrsim  v^{-(q + 1)} \\
			Y\psi &\gtrsim v^{-(q+1)}
			\end{aligned}
			\end{equation*}
			and an analogous estimate (with the same positive sign)\footnote{By an ``analogous"  estimate along the left event horizon we mean that we replace $\partial_t$ by $-\partial_t$ (recalling that $-\partial_t$ is future directed on the left event horizon!), $Y$ by $\tilde{Y}$, and $v$ by $u$. The rest remains unchanged. Here, for example, the analogous estimate would take the form $-\partial_t \psi \gtrsim u^{-(q+1)}$ and $\tilde{Y}\psi \gtrsim u^{-(q+1)}$.}  along the portion of the left event horizon. Moreover, assume that $\partial_u \psi >0$ and $\partial_v \psi >0$ hold along the portion of $\Sigma'$ in the black hole interior that connects the two event horizons and let $0< r_0 < 2m$. Then there exists a $C>0$ such that 
			\begin{equation*}
			\psi (t,r) \geq \psi(t,r_0) + C  (|t|+1)^{-(q+1)} \log \frac{r_0}{r}
			\end{equation*}
			holds in $\{0<r<r_0\}$.   
			
			\item Assume $\psi$ satisfies along the right event horizon $\Hp \cap \{v \geq 1\}$
			\begin{equation} \label{ThmSphSymAssumptionLowerBound}
			\partial_t \psi \gtrsim v^{-(q + 1)} \qquad \textnormal{ for } v \geq v_0 \;,
			\end{equation}
			where $v_0 \geq 1$ and $q >-1$; and similarly\footnote{Here we do not require that the analogous estimate has the same positive sign, i.e., we also allow for $-\partial_t \psi \lesssim -u^{-(q+1)}$.} for the left event horizon. Moreover, let $0< r_0 < 2m$. Then there exists a $t_0 \geq 1$ and a $C>0$ such that 
			\begin{equation}\label{Mihalis}
			\psi (t,r) \geq \psi(t,r_0) + C |t|^{-(q+1)}  \log \frac{r_0}{r} 
			\end{equation}
			holds in $\{0<r<r_0\}\cap\{ t\geq t_0\}$ -- and similarly\footnote{In the case of the assumption $-\partial_t \psi \lesssim -u^{-(q+1)}$ along the left event horizon, i.e., assuming a negative sign, \eqref{Mihalis} has to be changed to $\psi(t,r) \leq \psi(t,r_0) - C |t|^{-(q+1)} \log \frac{r_0}{r}$, i.e., the wave blows up to $-\infty$ for $r\to 0$. Cf.\ the discussion below.} for the vicinity of the left  end of the singular hypersurface $\{r = 0\}$.   
		\end{enumerate}
	\end{theorem}
	\begin{figure}[h]
		\centering
		\def\svgwidth{7cm}
		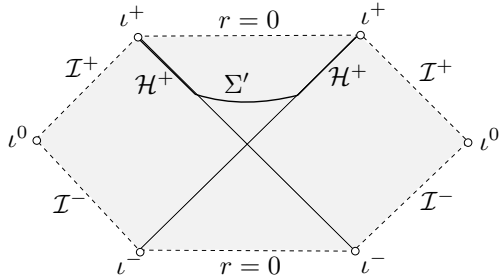
		\caption{The hypersurface $\Sigma'$ in the maximal analytic Schwarzschild spacetime} \label{FigSigmaPrime}
	\end{figure}
	We remark that replacing $\psi$ by $-\psi$ gives another version of this theorem with reversed signs on the assumptions, { i.e., $\partial_t\psi\lesssim -v^{-(q+1)}$ etc}. We would like to emphasise that the sign of the derivatives in the assumptions determines the sign of $\psi$ close to the singularity, {i.e., whether $\psi\rightarrow+\infty$ or $\psi\rightarrow-\infty$ as $r\rightarrow0$}. Thus, while it is important for the first part of this theorem that the signs on the left and right event horizon are the same in order to obtain blow up of $\psi$ on all of $\{r = 0\}$, the signs might very well be different in the second part of this theorem -- and thus leading to a blow-up of $\psi$ with different signs near the opposite endpoints of $\{r=0\}$. { In the latter case however, there will unavoidably be a non-empty subset of $\{r=0\}$ on which $\psi$ will be bounded; this follows from the continuity of the coefficient of its logarithmic part, see \eqref{AsymptoticExp} below}.
	
	Section \ref{SecExp} studies waves without any symmetry assumptions. The main theorem proved in this section provides in particular an expansion of the wave close to the singular hypersurface:
	\begin{theorem}
		\label{ThmExp}
		Let $\psi$ be a smooth solution to the wave equation \eqref{WaveEquation} on the maximal analytic Schwarzschild spacetime that satisfies the following bounds along the event horizon $\Hp \cap \{v \geq 1\}$
		\begin{equation}\label{ThmExpAssumptions}
		\begin{aligned} 
		|\psi| &\lesssim v^{-q}   \\ 
		| \partial_t^{(i)} \Omega_k^{(j)} \psi | &\lesssim v^{-(q + \delta)} \quad \textnormal{ for } i,j \in \N_0, 1 \leq i + j \leq 6, k = 1,2,3 \;,
		\end{aligned}
		\end{equation}
		where $q > 0$ and $ \delta  \geq 0$. Moreover, we assume that the analogous bounds are satisfied along the left event horizon.
		
		Then 
		there exists an $r_0 \in (0, 2m)$ close to $2m$ and
		there exist functions $A,B \in C^\infty(\R \times \mathbb{S}^2)$ with $|A(t, \omega)| \lesssim (|t|+1)^{-(q + \delta)}$ and $|B(t,\omega)| \lesssim (|t| +1)^{-q}$  and $P \in C^\infty(M)$ with $|P(t,r, \omega)| \lesssim r|\log r|(|t| +1)^{-q}$, such that the expansion
		\begin{equation}
		\label{AsymptoticExp}
		\psi(t,r, \omega) = A(t, \omega)\log r + B(t,\omega) + P(t,r,\omega)
		\end{equation}
		holds in $\{0< r< r_0\}$. 
		Moreover, the bound
		\begin{equation*}
		|\psi| \lesssim v^{-q} 
		\end{equation*}
		holds in $\{v \geq 1\} \cap \{r_0 \leq r \leq 2m\}$ (together with the analogous estimate for the left side of the black hole interior).
	\end{theorem}
	The proof of the preceding theorem is based on deriving energy estimates in physical space. First, we propagate energy bounds on the event horizon, which are derived from \eqref{ThmExpAssumptions}, to some constant $r=r_0$ hypersurface, using the redshift vector field of Dafermos and Rodnianski, \cite{DafRod09a}, \cite{DafRod08}. Then we derive suitable $r$-weighted renormalised energy estimates from $r=r_0$ to $r=0$ to obtain an upper bound on the asymptotic behaviour of general solutions. This defines the principal term in the expansion of $\psi$ via 
	\begin{align*}
	A:\overset{L^2}{=}\lim_{r\rightarrow0}\frac{\psi}{\log r}.
	\end{align*}
	Then the next order terms in the expansion of $\psi$ and the control on the error term are derived by successively subtracting the former order terms from $\psi$ and deriving energy estimates for each difference, which in particular imply (see Section \ref{SecExp}):
	\begin{align*}
	B:\overset{L^2}{=}\lim_{r\rightarrow0}(\psi-A\log r)
	\end{align*}
	and
	\begin{align*}
	P:=\psi-A\log r-B,&&|P|\lesssim r|\log r|(|t|+1)^{-q}.
	\end{align*}
	We would also like to draw the reader's attention to the different decay rates of $A$ and $B$ in $t$: $A$ inherits the faster decay rate in $v$ of the derivatives along the event horizon, while $B$ inherits the slower decay rate of the field itself along the event horizon. The reason for this is that the logarithmic blow-up is triggered and controlled by the first derivatives, cf.\ Theorem \ref{ThmSphSym} and the preceding discussion, as well as the proof of Theorem \ref{ThmSphSym}, in particular \eqref{ControlBlowUpFirstDer}.
	
	The following theorem is an easy consequence of  Theorem \ref{ThmSphSym} a) and Theorem \ref{ThmExp}.
	\begin{theorem}\label{ThmMaster}
		Let $\psi = \psi_{ 0} + \psi_{\ell \geq 1}$ be a smooth solution to the wave equation \eqref{WaveEquation} on the maximal analytic Schwarzschild spacetime, where $\psi_0$ denotes the spherically symmetric part of $\psi$ and $\psi_{\ell \geq 1}$ the projection on the higher spherical harmonics. Assume that the following bounds are satisfied along the right event horizon $\Hp \cap \{ v \geq 1\}$:
		\begin{equation} \label{ThmMasterCond}
		\begin{aligned}
		|\psi| &\lesssim v^{-q}   \\ 
		|\partial_t^{(i)} \Omega_k^{(j)} \psi | &\lesssim v^{-(q +1)} \quad &&\textnormal{ for }  i,j \in \N_0,  1 \leq i + j \leq 6,   k = 1,2,3 \\
		|\partial_t^{(i)} \Omega_k^{(j)} \psi_{\ell \geq 1} | &\lesssim v^{-(q +1 + \varepsilon)} \quad &&\textnormal{ for } i,j \in \N_0, 1 \leq i + j \leq 6,    k = 1,2,3 \\
		\big|\partial_t\psi_0\big| &\gtrsim v^{-(q+1)} &&\textnormal{ for } v \geq v_0  \;,
		\end{aligned}
		\end{equation}
		where $q, \varepsilon>0$, $v_0 \geq 1$. 
		Moreover, we assume that the analogous estimates hold along the left event horizon.  Then there exists a $t_0 > 0$ such that the conclusion of Theorem \ref{ThmExp} holds with $A = A_0 + A_{\ell \geq 1}$, $|A_0| \sim (|t|+1)^{-(q+1)}$ for $|t| \geq t_0$, and $|A_{\ell \geq 1}| \lesssim (|t|+1)^{-(q+1+\varepsilon)}$. Here, $A_0$ is the spherically symmetric part of $A$. In particular, it follows that $|A| \sim (|t|+1)^{-(q+1)}$ for $t$ large enough so that $\psi$ blows up pointwise and logarithmically in a neighbourhood of the endpoints $|t| = \infty$ of the singular hypersurface $\{r=0\}$.
	\end{theorem}
	
	\begin{proof}
		By Theorem \ref{ThmExp}  it follows from the first two conditions in \eqref{ThmMasterCond} that there exists an $r_0 \in (0,2m)$ and functions $A,B \in C^\infty(\R \times \mathbb{S}^2)$ with $|A(t, \omega)| \lesssim (|t|+1)^{-(q + 1)}$ and $|B(t,\omega)| \lesssim (|t|+1)^{-q}$, and $P \in C^\infty(M)$ with $P = \mathcal{O}(r|\log r|(|t|+1)^{-q})$ such that the expansion \eqref{AsymptoticExp} holds in $\{0< r< r_0\}$. We decompose $A$ in its spherically symmetric part $A_0(t) = \int_{\mathbb{S}^2} A(t,\theta,\varphi) \, \sin \theta \, d\theta d\varphi$ and in the reminder $A_{\ell \geq 1} = A - A_0$. It follows immediately that $|A_0| \lesssim (|t|+1)^{-(q+1)}$. Let us now assume without loss of generality that $\partial_t \psi_0$ is positive in the  last assumption in \eqref{ThmMasterCond} along the right event horizon (otherwise replace $\psi$ by $-\psi$). It then follows from Theorem \ref{ThmSphSym} b) that there exists a $C> 0$ such that 
		\begin{equation}\label{LB}
		\psi_0(t,r) \geq \psi_0(t,r_0) + C t^{-(q+1)} \log \frac{r_0}{r}
		\end{equation}
		holds for $t$ large enough and $0<r<r_0$. Moreover, considering the spherically symmetric part of the expansion \eqref{AsymptoticExp} gives
		\begin{equation}\label{SE}
		\psi_0 (t,r) = A_0(t) \log r + B_0(t) + P_0(t,r) 
		\end{equation}
		in the same region. Combining \eqref{LB} and \eqref{SE}, dividing by $\log r$ and letting $r$ go to $0$ gives  $A_0(t) \leq - C t^{-(q+1)}$. Together with the same argument for the left event horizon this shows $|A_0| \gtrsim (|t|+1)^{-(q+1)}$ for $|t|$ large enough, and hence $|A_0| \sim (|t|+1)^{-(q+1)}$ for $|t|$ large enough. Finally we note that the first condition of \eqref{ThmMasterCond} also implies $|\psi_0| \lesssim v^{-q}$ along the event horizon, and thus the same bound also holds for $|\psi_{\ell \geq 1}|$. Together with the third condition of \eqref{ThmMasterCond} it follows from Theorem \ref{ThmExp} that $|A_{\ell \geq 1}| \lesssim (|t|+1)^{-(q+1+\varepsilon)}$ holds.
	\end{proof}

	The methods developed in \cite{AnArGa16a} can be used to show that the assumptions \eqref{ThmMasterCond} of the above theorem are satisfied for solutions arising from generic initial data prescribed on $\Sigma$, see the forthcoming work \cite{AnArGa18} (but also \cite{AnArGa16b}):
	\begin{theorem}[Angelopoulos--Aretakis--Gajic \cite{AnArGa16a,AnArGa16b,AnArGa18}] \label{ThmAAG}
		A solution of the wave equation \eqref{WaveEquation} on the maximal analytic Schwarzschild spacetime arising from  generic smooth initial data on $\Sigma$ with a certain upper bound on the decay towards spacelike infinity $\iota^0$ satisfies the conditions \eqref{ThmMasterCond} in Theorem \ref{ThmMaster} along the right (and left) event horizon $\Hp \cap \{v \geq 1\}$ with $v_0$ sufficiently large, $\varepsilon = 2$, and $q = 2$ or $q = 3$ depending on the upper bound on the decay towards $\iota^0$.
	\end{theorem}
\begin{remark}
The notion of genericity of smooth solutions in the previous theorem and in the subsequent corollary comes from \cite{AnArGa16a,AnArGa16b,AnArGa18}. The initial data on $\Sigma$ are assumed to have finite energies of the form \eqref{EnergyNorm}, for a certain finite number of derivatives and $r$-weights. Then, the case $q=2$ holds true when the Newman-Penrose constant of $\psi$, $I_0[\psi]$, is non-zero. On the other hand, the case $q=3$ corresponds to initial data with $I_0[\psi]=0$, e.g., compactly supported, for which the Newman-Penrose constant $I_0^{(1)}[\psi]$ of the time integral of $\psi$, see $(1.12)$ in \cite{AnArGa16b}, is non-zero.
\end{remark}
	In particular, we can conclude from the above two theorems the following
	\begin{corollary}
		Smooth solutions to the wave equation $\square_g\psi=0$ in the maximal analytic Schwarzschild spacetime, arising from generic initial data prescribed on $\Sigma$, blow up pointwise and logarithmically in a neighbourhood of the endpoints $|t|=+\infty$ of the singular hypersurface $\{r=0\}$ in the interior of the black hole. 
	\end{corollary}
	Note that the above corollary only ensures blow-up of generic solutions to the wave equation near the endpoints $|t| = \infty$ of the singular hypersurface $\{r=0\}$. However, we do have the following 
	\begin{theorem}\label{ThmOpenSet}
		There exists an open set of initial data on $\Sigma$ for the wave equation \eqref{WaveEquation}  on the maximal analytic Schwarzschild spacetime such that the arising solutions blow up logarithmically everywhere on the singular hypersurface $\{r=0\}$.
	\end{theorem}
	Let us add here that we consider smooth initial data on $\Sigma$ such that an energy quantity is finite, see Section \ref{SecOpenSet} for more details.
	The openness of the initial data is measured with respect to the topology induced by this energy.
	
	The proof of Theorem \ref{ThmOpenSet} is given in Section \ref{SecOpenSet}. It proceeds by first showing the existence of a spherically symmetric solution such that the assumptions of Theorem \ref{ThmSphSym} a) are satisfied, and hence the solution blows up on all of the singular hypersurface $\{r= 0\}$. The second step makes use of the fact that small perturbations of the initial data of this spherically symmetric solution also lead to small perturbations of $A$ in the expansion \ref{AsymptoticExp} -- so that small enough perturbations still blow up on all of $\{r=0\}$.
	
	In Section \ref{sec:fullexp} we give a proof for the full asymptotic expansion of linear waves in the interior. This time we use a different argument, based on considering the wave equation as an ODE in $r$ and treating the spatial derivatives as error terms. It is simpler than using a separate energy argument for each order in the expansion, as in the proof of (\ref{AsymptoticExp}), but at the same time more wasteful with derivatives.
	\begin{theorem}\label{ThmFullExp}
		Let $\psi$ be a smooth solution to the wave equation in the interior region $\{0<r< 2m\}$. Then given $N\in\mathbb{N}$, $N\ge0$, $\psi$ can be represented by the following expansion:
		\begin{align}\label{fullasymexp}
		\psi(r,t,\omega)=\sum_{n=0}^N\zeta_n(t,\omega)r^n\log r
		+\sum_{n=0}^N\eta_n(t,\omega)r^n+R_N(t,r,\omega)
		\end{align}
		where $R_N$ satisfies $|\partial_t^{(i)}\Omega_k^{(j)}R_N|\lesssim r^{N+1}|\log r|$, $|\partial_r\partial_t^{(i)}\Omega_k^{(j)}R_N|\lesssim r^N|\log r|$, for all $i,j$, and $\zeta_n,\eta_n$ are smooth functions given by the recurrence relations:
		\begin{align}\label{zetanetan}
		\notag\zeta_{n+1}:=&\,\frac{n(n+1)\zeta_n+\Delta_{\mathbb{S}^2}\zeta_n+\sum_{l=0}^{n-3}(\frac{1}{2m})^l\frac{\partial_t^2\zeta_{n-3-l}}{2m}}{2m(n+1)^2},\\
		\eta_{n+1}:=&\,\frac{(2n+1)\zeta_n+n(n+1)\eta_n-2m(2n+2)\zeta_{n+1}+\Delta_{\mathbb{S}^2}\eta_n+\sum_{l=0}^{n-3}(\frac{1}{2m})^l\frac{\partial_t^2\eta_{n-3-l}}{2m}}{2m(n+1)^2},
		\end{align}
		for every $n=0,\ldots,N-1$, where $\zeta_0=A,\eta_0=B$ and by convention $\zeta_l=\eta_l=0$, for $l<0$.
	\end{theorem}
	Finally, in Section \ref{isoS} we first demonstrate via a backwards construction argument that for any pair of smooth functions $A(t,\omega),B(t,\omega)$ with certain decay in $t$, there exists a smooth solution $\psi$ to the wave equation in the {\it interior} with the same decay in $t$, having the asymptotic expansion (\ref{AsymptoticExp}) and hence generating numerous blow-up examples, see Theorem \ref{backthm}. In fact, we subsequently show that the map $\psi\leftrightarrow(A,B)$ induced by our forwards-backwards analysis is an isomorphism.
	\begin{theorem}\label{ThmIso}
		The map
		\begin{align*}
		\mathcal{S}:\{\psi \;\mathrm{smooth}, \square_g\psi=0\;\mathrm{in}\;\{0<r\leq r_0\},\;|\partial_t^{(i)}\Omega^{(j)}_k\psi(r=r_0)|\lesssim (1+|t|)^{-q},\;k=1,2,3, \;\mathrm{for \;all}\;i,j\}\\
		\notag\to\{A(t,\omega),B(t,\omega)\;\mathrm{smooth},\;|\partial_t^{(i)}\Omega_k^{(j)}A|,|\partial_t^{(i)}\Omega_k^{(j)}B|\lesssim(1+|t|)^{-q},\;k=1,2,3,\;\mathrm{for\;all}\;i,j\}
		\end{align*}
		given by 
		\begin{align}\label{S}
		\mathcal{S}(\psi)=\big(\lim_{r\rightarrow0}\frac{\psi}{\log r}, \lim_{r\rightarrow0} [\psi-(\lim_{r\rightarrow0}\frac{\psi}{\log r})\log r]\big)=(A,B)
		\end{align}
		is an isomorphism, for any $0<r_0<2m$.
	\end{theorem}
	We remark that there is a derivative loss in the energy estimates in both directions of our forwards-backwards analysis and that is precisely the reason why we define the map $\mathcal{S}$ above between sets of smooth functions and not energy spaces containing finitely many derivatives. 
	
	We conclude the introduction with the following remark:
	\begin{remark}\label{rem:symexp}
		In the canonical $(t,r)$-coordinates, the wave equation \eqref{WaveEquation} under spherical symmetry reads\footnote{See also \eqref{sswaveeq} in the next section.}
		\begin{equation*}
		-(1 - \frac{2m}{r})^{-1} \partial_t^2 \psi + (1 - \frac{2m}{r})\partial^2_r\psi + \frac{2}{r}(1 - \frac{m}{r}) \partial_r \psi = 0 \;.
		\end{equation*}
		Looking for solutions independent of $t$, the equation takes on the form
		\begin{equation} \label{Frob}
		\frac{d^2}{dr^2} \psi (r) + \frac{2(r-m)}{r(r-2m)} \frac{d}{dr} \psi (r) = 0 \;.
		\end{equation}
		It now follows from the Frobenius method that \eqref{Frob} has a fundamental system of solutions given by $$\{h_1(r), h_2(r) + c (\log r) \cdot h_1(r)\},$$ where $h_1(r)$ and $h_2(r)$ are analytic near $r=0$ and $c \neq 0$. It is easy to see that one can choose $h_1(r) =1$. This does not only demonstrate the logarithmic blow-up, but also motivates the form of the asymptotic expansion \eqref{AsymptoticExp} of a general solution (without symmetry assumptions) in Theorem \ref{ThmExp} and of \eqref{fullasymexp} in Theorem  \ref{ThmFullExp}. 
	\end{remark}

	\section{Spherically symmetric waves in the black hole interior: proof of Theorem \ref{ThmSphSym}}\label{SecSphericalSymmetry}

	This section provides the proof of Theorem \ref{ThmSphSym}. We start with the proof of part b) of Theorem \ref{ThmSphSym} and establish a series of results which, when put together, will accomplish the proof. The proof of part a) of Theorem \ref{ThmSphSym} follows thereafter easily from the methods developed.
	
	Assuming the wave $\psi$ to be spherically symmetric, the wave equation \eqref{WaveEquation} reads in  the $(u,v)$-coordinates introduced in Section \ref{SecNotation}
	\begin{equation}\label{sswaveeq}
	\partial_u\partial_v\psi+\frac{\partial_ur}{r}\partial_v\psi+\frac{\partial_vr}{r}\partial_u\psi=0 \;.
	\end{equation}
	It is easily seen that the spherically symmetric wave equation \eqref{sswaveeq} is equivalent to
	\begin{equation}
	\label{WaveEq2}
	\partial_u (r \partial_v \psi) = -\frac{\partial_v r}{r} r \partial_u \psi
	\end{equation}
	as well as to
	\begin{equation}
	\label{WaveEq3}
	\partial_v (r \partial_u \psi) = -\frac{\partial_u r}{r} r \partial_v \psi \;.
	\end{equation}
	Let us remark here that the monotonicity property for spherically symmetric waves in the black hole interior, mentioned in the beginning of Section \ref{SecMainThm}, can now be easily understood from \eqref{WaveEq2} and \eqref{WaveEq3} by virtue of $\partial_v r$ and $ \partial_u r$ being negative in the black hole interior.
	
	We now begin by showing that the lower bound \eqref{ThmSphSymAssumptionLowerBound} in Theorem \ref{ThmSphSym} b) also implies eventually a lower bound for $Y\psi$ along the event horizon $\Hp$. This is a manifestation of the red-shift effect on $\Hp$. In order to ease notation, we set $p = q+1$ in this section. We also remind the reader that along the right event horizon we have $\partial_v = \partial_t$ and that $\partial_v$ in $(u,v)$-coordinates extends continuously to $\partial_v$ in $(v,r)$-coordinates on the right event horizon. Thus the notation $\partial_v$ on $\mathcal{H}^+$ is unambiguous.
	
	\begin{proposition} \label{PropRedShiftHorizon}
		Assume $\psi$ is a smooth solution to the spherically symmetric wave equation \eqref{sswaveeq} which satisfies  $\partial_v\psi|_{\Hp} \geq C_0 v^{-p}$ along the event horizon for $v \geq v_0$, where $v_0 \geq 1$ is some constant and $p>0$. For any $\varepsilon >0$ there then exists a $v_1 \geq 1$ such that
		\begin{equation*}
		Y\psi|_{\Hp} (v) \geq (2C_0 - \varepsilon) v^{-p}
		\end{equation*}
		holds along the event horizon for $v \geq v_1$.
	\end{proposition}
	
	The following lemma is needed for the proof. 
	
	\begin{lemma}
		\label{AuxLemma}
		Let $p, D, v_0 >0$ be constants. For $v \geq v_0$  the following holds:
		\begin{equation*}
		e^{-Dv}\int\limits_{v_0}^v e^{Dv'} (v')^{-p} \, dv' = \frac{1}{D} v^{-p} + \mathcal{O}(v^{-(p+1)}) \;.
		\end{equation*}
	\end{lemma}
	
	\begin{proof}[Proof of Lemma \ref{AuxLemma}:]
		We compute
		\begin{equation}
		\label{FirstIntParts}
		\begin{split}
		e^{-Dv} \int\limits_{v_0}^v e^{Dv'} (v')^{-p} \, dv' &= e^{-Dv}\frac{1}{D} \int\limits_{v_0}^v \partial_{v'} \Big[ e^{Dv'} (v')^{-p}\Big] \,dv' + e^{-Dv} \frac{p}{D} \int\limits_{v_0}^v  e^{Dv'} (v')^{-(p+1)}\,dv' \\
		&=\frac{1}{D} v^{-p} - \frac{1}{D} e^{D(v_0 -v)} v_0^{-p} + \frac{p}{D} e^{-Dv} \int\limits_{v_0}^v  e^{Dv'} (v')^{-(p+1)}\,dv' \;.
		\end{split}
		\end{equation}
		Similarly, we have
		\begin{equation}
		\label{AuxEst}
		\frac{p}{D} e^{-Dv} \int\limits_{v_0}^v  e^{Dv'} (v')^{-(p+1)}\,dv' = \frac{p}{D^2} v^{-(p+1)} - \underbrace{\frac{p}{D^2} e^{D(v_0 - v)} v_0^{-(p+1)}}_{\geq 0} + \frac{p(p+1)}{D^2} e^{-Dv}\int\limits_{v_0}^v e^{Dv'} (v')^{-(p+2)} \, dv' \;.
		\end{equation}
		We can now choose $v_1(p,D) >0$ big enough such that $\frac{p+1}{D} (v')^{-(p+2)} \leq \frac{1}{2}(v')^{-(p+1)}$ holds for all $v' \geq v_1$. To show that the left hand side of \eqref{AuxEst} is $\mathcal{O}(v^{-(p+1)})$, we bring one half of it over to the right hand side and drop the second, negative term of \eqref{AuxEst}:  
		\begin{equation*}
		\begin{split}
		0 \leq \frac{p}{2D} e^{-Dv} \int\limits_{v_0}^v  e^{Dv'} (v')^{-(p+1)}\,dv' &\leq \frac{p}{D^2} v^{-(p+1)} + \frac{p(p+1)}{D^2} e^{-Dv}\int\limits_{v_0}^{v_1} e^{Dv'} (v')^{-(p+2)} \, dv'\\ &+ \underbrace{\frac{p(p+1)}{D^2} e^{-Dv}\int\limits_{v_1}^v e^{Dv'} (v')^{-(p+2)} \, dv' - \frac{p}{2D} e^{-Dv} \int\limits_{v_0}^v  e^{Dv'} (v')^{-(p+1)}\,dv' }_{\leq 0} \;.
		\end{split}
		\end{equation*}
		Hence, the last term in \eqref{FirstIntParts} is $\mathcal{O}(v^{-(p+1)})$, which concludes the proof.
	\end{proof}
	
	\begin{proof}[Proof of Proposition \ref{PropRedShiftHorizon}:]
		Manipulating the form \eqref{WaveEq2} of the spherically symmetric wave equation we obtain
		\begin{equation*}
		\begin{split}
		\partial_v \Big( - r\frac{\partial_u \psi}{\partial_u r}\Big) &= -\frac{1}{\partial_u r} \partial_v ( r \partial_u \psi) + \frac{ r \partial_u \psi}{(\partial_u r)^2} \partial_v \partial_u r \\
		&= \partial_v \psi + \frac{2m}{r \partial_u r} \partial_u \psi \;,
		\end{split}
		\end{equation*}
		where we have used $\partial_v \partial_u r = \frac{2m}{r^2} \partial_v r = \frac{2m}{r^2} \partial_u r$. Using the definition \eqref{DefY} of the vector field $Y$ this reads
		\begin{equation}
		\label{EqForY}
		\partial_v\Big( \frac{r}{2} Y\psi\Big) = \partial_v \psi - \frac{m}{r^2}\Big(r Y \psi\Big) \;.
		\end{equation}
		We now recall that the vector field $\partial_v$ in $(u,v)$-coordinates extends continuously to the event horizon and equals there $\partial_v$ in $(v,r)$-coordinates (which again equals $\partial_t$). It is convenient not to change the $(u,v)$-coordinate system when computing at the event horizon $\Hp$ but to include $\Hp$ as the asymptotic hypersurface $u = -\infty$, as is done in particular in the next proposition.
		
		We now consider \eqref{EqForY} at the event horizon $\Hp$:
		\begin{equation*}
		\partial_v (Y\psi)|_{\Hp} = \frac{1}{m} \partial_v \psi|_{\Hp}  - \frac{1}{2m} Y\psi|_{\Hp} \;,
		\end{equation*}
		which is equivalent to
		\begin{equation*}
		\partial_v\Big(e^{\frac{v}{2m}} Y\psi|_{\Hp} \Big) = \frac{1}{m} e^{\frac{v}{2m}}\partial_v \psi|_{\Hp} \;.
		\end{equation*}
		Integrating from $v_0$ to $v$ and using Lemma \ref{AuxLemma} we thus obtain
		\begin{equation}\label{LowBoundV1}
		\begin{split}
		Y\psi|_{\Hp} (v) &= e^{\frac{1}{2m}(v_0 -v)} Y\psi|_{\Hp} (v_0) + \frac{1}{m} e^{-\frac{v}{2m}} \int\limits_{v_0}^v e^{\frac{v'}{2m}} \partial_{v'} \psi|_{\Hp} (v')\, dv' \\
		&\geq e^{\frac{1}{2m}(v_0 -v)} Y\psi|_{\Hp} (v_0) + \frac{1}{m} e^{-\frac{v}{2m}} \int\limits_{v_0}^v e^{\frac{v'}{2m}} C_0 (v')^{-p} \, dv' \\
		&= e^{\frac{1}{2m}(v_0 -v)} Y\psi|_{\Hp} (v_0) + 2C_0 v^{-p} + \mathcal{O}(v^{-(p+1)}) \;.
		\end{split}
		\end{equation}
	\end{proof}
	\begin{remark}[to Proposition \ref{PropRedShiftHorizon}:]
		Let us remark that for given $\varepsilon >0$, the lower bound on the affine time $v_1$, provided by \eqref{LowBoundV1}, depends continuously on $v_0, Y\psi|_{\Hp}(v_0), C_0$ and the parameter $m$ of the black hole.
	\end{remark}
	
	The next proposition shows that the lower bounds on $\partial_v \psi$ and $Y\psi$ on $\Hp \cap \{v \geq v_1\}$ can be propagated slightly off the event horizon into the black hole interior.
	
	\begin{proposition} \label{PropOffHp}
		Let $0< r_0 < 2m$ and $D_0 := \frac{2m}{r_0^2}$. Assume $\psi$ is a smooth solution to the spherically symmetric wave equation that satisfies $r\partial_v\psi|_{\Hp} \geq 2C_1 v^{-p}$ and $rY\psi|_{\Hp} (v) \geq \frac{2C_1}{mD_0}  v^{-p}$ along the event horizon for $v \geq v_1$, where $v_1 \geq 1$ and $C_1 >0$ are constants.  Then there exists a $- \infty < u_0$ such that
		\begin{equation*}
		r\partial_v \psi \geq C_1 v^{-p} \quad \textnormal{ and } \quad rY\psi \geq \frac{C_1}{mD_0} v^{-p}
		\end{equation*}
		holds in $[-\infty,u_0] \times [v_1, \infty) \cap \{r \geq r_0\}$.
	\end{proposition}
	
	\begin{proof}
		Let $-\infty < u_0$ be such that $rY\psi \geq \frac{C_1}{mD_0} v^{-p}$ and $r\partial_v \psi \geq C_1 v^{-p}$ hold on $[-\infty, u_0] \times \{v_1\}$.
		Let 
		\begin{equation*}
		J :=\Big{\{} v \in [v_1, \infty) \; \Big| \; rY\psi \geq \frac{C_1}{mD_0} v^{-p} \textnormal{ and } r\partial_v \psi \geq C_1 v^{-p} \quad \textnormal{ in } \big([-\infty, u_0] \times [v_1, v]\big) \cap \{r \geq r_0\} \Big{\}} \;.
		\end{equation*}
		The interval $J$ is clearly non-empty and closed. We show openness, from which the proposition follows. 
		
		Let $v \in J$ and $u \in [-\infty, u_0]$ with $r(v,u) \geq r_0$. In $\big([-\infty, u_0] \times [v_1,v]\big) \cap \{r \geq r_0\}$ we in particular have $rY\psi \geq 0$, so that we obtain from \eqref{EqForY} that
		\begin{equation*}
		\partial_v(rY\psi) \geq 2\partial_v \psi - D_0 rY\psi
		\end{equation*}
		holds in this region, which is equivalent to
		\begin{equation*}
		\partial_v \Big(e^{D_0 v}rY\psi\Big) \geq e^{D_0 v}2\partial_v \psi \;.
		\end{equation*}
		Integrating from $v_1$ to $v$ yields
		\begin{equation*}
		\begin{split}
		rY\psi(v,u) &\geq e^{D_0(v_1 - v)}rY\psi(v_1,u) + e^{-D_0v} \int\limits_{v_1}^v e^{D_0 v'} 2\partial_v \psi (v',u)\, dv'\\
		&\geq e^{D_0(v_1 - v)}rY\psi(v_1,u) + e^{-D_0v} \int\limits_{v_1}^v e^{D_0 v'} \frac{C_1}{m} (v')^{-p} \, dv' \\
		&\geq e^{D_0(v_1 - v)}rY\psi(v_1,u) + \frac{C_1}{mD_0} v^{-p} - \frac{C_1}{mD_0} v_1^{-p} e^{D_0(v_1 - v)} + \underbrace{\frac{pC_1}{mD_0} e^{-D_0 v} \int\limits_{v_1}^v e^{D_0 v'} (v')^{-(p+1)} \, dv'}_{>0} \\
		&> \frac{C_1}{mD_0} v^{-p} \;,
		\end{split}
		\end{equation*}
		where, to obtain the third inequality, we used the same integration by parts computation as in \eqref{FirstIntParts}, and to obtain the final inequality, we used $rY\psi (v_1,u) \geq \frac{C_1}{mD_0} v_1^{-p}$.
		Moreover, from \eqref{WaveEq2} we obtain
		\begin{equation*}
		\begin{split}
		r\partial_v\psi (v,u) &= r\partial_v \psi(v, -\infty) + \underbrace{\int\limits_{-\infty}^u - \frac{\partial_v r}{r} r\partial_u \psi (v, u') \, du'}_{> 0} \\
		&\geq 2C_1 v^{-p} \\
		&> C_1 v^{-p} \;.
		\end{split}
		\end{equation*}
		Together with the compactness of $\big([-\infty, u_0] \times \{v\}\big) \cap \{r \geq r_0\}$ this shows openness of $J$.
	\end{proof}
	
	The next proposition propagates the lower decay bounds on $\partial_u \psi$ and $\partial_v \psi$ in $v$ all the way to the singularity and at the same time shows that they  blow up there like $\sim \frac{1}{r^2}$.
	
	\begin{proposition}
		\label{ToSing}
		Assume $\psi$ is a smooth solution to the spherically symmetric wave equation that satisfies $r\partial_v\psi > \frac{C_2}{r} v^{-p}$ and $r\partial_u\psi  > \frac{C_2}{r} v^{-p}$ on $\{r=r_0\} \cap \{v \geq v_1\}$, where $0<r_0 < 2m$, $v_1 \geq 1$ and $C_2 >0$ are constants. Let $\Delta t := 2[r^*(0) - r^*_0]$, where $r^*_0 = r^*(r_0)$.
		
		Then
		\begin{equation*}
		r\partial_v \psi \geq \frac{C_2}{r} v ^{-p} \quad \textnormal{ and } \quad r\partial_u\psi \geq \frac{C_2}{r} v^{-p}
		\end{equation*}
		hold in $\{0 < r \leq r_0\} \cap \{ v \geq v_1 + \Delta t\}$.
	\end{proposition}

	\begin{proof}
		Let $u_0 =  2r^*_0 - v_1$ and take $v_2 \geq v_1$. We consider the region $D(v_2) := \{0<r\leq r_0\} \cap \{ u \leq 2r^*_0 - v_2\} \cap \{v \leq v_2 + \Delta t\}$, see also Figure \ref{FigRegionsProp}. Recall that $r^* = \frac{1}{2}(v+u)$. Moreover, $\frac{1}{2}(v_2 + \Delta t + 2r_0^* - v_2) = r^*(0)$, and thus the hypersurfaces $u = 2r^*_0 - v_2$ and $v = v_2 + \Delta t$ `intersect at $r=0$' in the Penrose diagram.
		
		\begin{figure}[h] 
			\centering
			\def\svgwidth{7cm}
			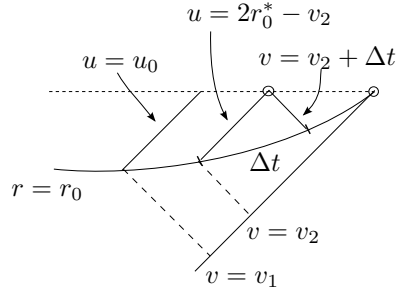
			\caption{The regions in the proof of Proposition \ref{ToSing}} \label{FigRegionsProp}
		\end{figure}
		
		Let
		\begin{equation*}
		J:= \Big{\{} r' \in (0,r_0] \, \Big| \, r\partial_u \psi \geq \frac{C_2}{r} (v_2 + \Delta t)^{-p} \; \textnormal{ and }  r\partial_v \psi \geq \frac{C_2}{r} (v_2 + \Delta t)^{-p} \textnormal{ holds in } D(v_2) \cap \{r' \leq r \leq r_0\} \Big{\}} \;.
		\end{equation*}
		The interval $J$ is clearly non-empty and closed. We show openness. Let $(u,v)$ be such that $r(u,v) \in J$.

		Let $v_{r_0}(u)$ be such that $r\big(u, v_{r_0}(u)\big) = r_0$.
		It follows from \eqref{WaveEq3} that
		\begin{equation*}
		\begin{split}
		r \partial_u\psi(u,v) &= r\partial_u \psi (u,v_{r_0}(u)) + \int\limits_{v_{r_0}(u)}^v -\frac{\partial_u r}{r} r\partial_v\psi(u,v') \, dv' \\
		&\geq r\partial_u \psi (u,v_{r_0}(u)) + \int\limits_{v_{r_0}(u)}^v -\frac{\partial_v r}{r} \frac{C_2}{r} (v_2 + \Delta t)^{-p} (u,v')\, dv' \\
		&=  r\partial_u \psi (u,v_{r_0}(u)) + C_2 (v_2 + \Delta t)^{-p}  \int\limits_{v_{r_0}(u)}^v \partial_v (\frac{1}{r}) (u,v') \, dv' \\
		&=  r\partial_u \psi (u,v_{r_0}(u)) + C_2 (v_2 + \Delta t)^{-p} \frac{1}{r} (u,v) - C_2 (v_1 + \Delta t)^{-p} \frac{1}{r_0} \\
		&>  C_2 (v_2 + \Delta t)^{-p} \frac{1}{r} (u,v)  \;,
		\end{split}
		\end{equation*}
		where we have used the given assumption of the proposition for the last inequality.
		Now, let $u_{r_0}(v)$ be such that $r\big(u_{r_0}(v), v\big) = r_0$. It follows from \eqref{WaveEq2} that
		\begin{equation*}
		\begin{split}
		r \partial_v\psi(u,v) &= r\partial_v \psi (u_{r_0}(v), v) + \int\limits_{u_{r_0}(v)}^u -\frac{\partial_v r}{r} r\partial_u\psi(u',v) \, du' \\
		&\geq r\partial_v \psi (u_{r_0}(v), v) + \int\limits_{u_{r_0}(v)}^u -\frac{\partial_u r}{r} \frac{C_2}{r} (v_1 + \Delta t)^{-p} (u',v)\, du' \\
		&=  r\partial_v \psi (u_{r_0}(v), v) + C_2 (v_1 + \Delta t)^{-p}  \int\limits_{u_{r_0}(v)}^u  \partial_u (\frac{1}{r}) (u',v) \, du' \\
		&=  r\partial_v \psi (u_{r_0}(v), v)  + C_2 (v_1 + \Delta t)^{-p} \frac{1}{r} (u,v) - C_2 (v_1 + \Delta t)^{-p} \frac{1}{r_0} \\
		&>  C_2 (v_1 + \Delta t)^{-p} \frac{1}{r} (u,v)  \;,
		\end{split}
		\end{equation*}
		where we made again use of the given assumption of the proposition. Together with the compactness of $\{r = const\} \cap D(v_2)$, this establishes the openness of $J$ and thus shows $J = (0,r_0]$. Given now a point $(u,v) \in \{0 < r \leq r_0\} \cap \{ v \geq v_1 + \Delta t\}$, we have $(u,v) \in D(v-\Delta t)$. This finishes the proof of Proposition \ref{ToSing}.
	\end{proof}
	
	We are now ready to prove Theorem \ref{ThmSphSym} b):
	
	\begin{proof}[Proof of Theorem \ref{ThmSphSym} b):]
		It follows from Propositions \ref{PropRedShiftHorizon}, \ref{PropOffHp}, and \ref{ToSing} that for given $r_0 \in (0, 2m)$ there exists a $v_2 \geq 1$ and a $C>0$ such that 
		\begin{equation*}
		\partial_v\psi \geq \frac{C}{r^2}v^{-p} \quad \textnormal{ and } \quad \partial_u\psi \geq \frac{C_3}{r^2}v^{-p} 
		\end{equation*}
		holds in $\{0<r< r_0\} \cap \{v \geq v_2\}$. Recalling 
		\begin{equation*}
		\frac{\partial}{\partial r}\Big|_t = \frac{1}{1-\frac{2m}{r}}\Big(\frac{\partial}{\partial u}\Big|_v + \frac{\partial}{\partial v}\Big|_u\Big)\;,
		\end{equation*}
		it follows that there exists a $t_0 \geq 1$ and a  $C>0$ (possibly different from the previous constant) such that
		\begin{equation*}
		-\partial_r \psi \ge \frac{C}{r(2m-r)}t^{-p}
		\end{equation*}
		holds in $\{t_0 \leq t\} \cap \{0< r \leq r_0\}$.
		
		We  now compute for $t \geq t_0$ and $0<r< r_0$
		\begin{equation}\label{ControlBlowUpFirstDer}
		\begin{aligned}
		\psi(t,r )&= \psi(t,r_0) + \int\limits_{r_0}^{r} \partial_r \psi (t,r') \, dr' \\
		&\geq \psi(t,r_0)  + \int\limits_{r_0}^{r} \frac{C}{r'(r' - 2m)}t^{-p} \, dr' \\
		&\geq \psi(t,r_0) + \frac{C}{2m} t^{-p} \int\limits_{r}^{r_0} \frac{1}{r'}  \,dr' \\
		&= \psi(t,r_0) + \frac{ C}{2m } t^{-p}  \log \frac{r_0}{r} \;,
		\end{aligned}
		\end{equation}
		which concludes the proof of Theorem \ref{ThmSphSym} b).
	\end{proof}
	
	The proof of part a) of Theorem \ref{ThmSphSym} is very similar:
	
	\begin{proof}[Proof of Theorem \ref{ThmSphSym} a):]
		It follows from Proposition \ref{PropOffHp}, together with the basic monotonicity argument showing that positive signs of $\partial_v \psi$ and $\partial_u \psi$ are preserved in future development in the interior (cf.\ beginning of this section and the proof of Proposition \ref{ToSing}), that there exists an $r_0 \in (0,2m)$ and a $C>0$  such that 
		\begin{equation*}
		r \partial_v \psi \geq \frac{C}{r} (|t| + 1)^{-p} \qquad \textnormal{ and } \qquad r \partial_u \psi \geq \frac{C}{r} (|t| + 1)^{-p}
		\end{equation*}
		holds along $\{r = r_0\}$. The conclusion then follows from Proposition \ref{ToSing} and the argument given in the proof of Theorem \ref{ThmSphSym} b).
	\end{proof}

	\section{Expansion of waves near the singular hypersurface $\{r=0\}$: proof of Theorem \ref{ThmExp}}\label{SecExp}
	
	This section provides the proof of Theorem \ref{ThmExp}. In Section \ref{SecHR} we begin by propagating the assumed decay on the horizon, written in terms of energy norms, to a surface of constant $r$. Section \ref{SecRS} establishes the asymptotic expansion \eqref{AsymptoticExp} near $\{r= 0\}$ and propagates the obtained decay on a surface of constant $r$ all the way to the singular hypersurface. The proofs are based on energy estimates, the basics of which we recall in the following.
	
	Given a smooth function $\psi$ we define the \emph{stress-energy tensor} of $\psi$ to be
	\begin{equation*}
	T_{\mu \nu}(\psi) = \partial_\mu \psi \partial_\nu \psi - \frac{1}{2} g_{\mu \nu} \partial^\alpha \psi \partial_\alpha \psi \;.
	\end{equation*}
	The divergence of the stress-energy tensor equals 
	\begin{equation*}
	\nabla^\mu T_{\mu \nu}(\psi)=\partial_\nu\psi\cdot\square \psi
	\end{equation*}
	and so for solutions to the wave equation, $T_{\mu\nu}(\psi)$ is divergence free.
	For a vector field $X$ we define the associated current 
	\begin{equation*}
	J^X_\mu(\psi) = X^\nu T_{\mu \nu}(\psi) 
	\end{equation*}
	and set
	\begin{equation*}
	K^X(\psi) = \pi^X_{\mu \nu} T^{\mu \nu}(\psi)+X\psi\cdot\square\psi \;,
	\end{equation*}
	where $\pi^X_{\mu \nu} = \frac{1}{2}(\nabla_\mu X_\nu + \nabla_\nu X_\mu)$ is the \emph{deformation tensor} of $X$. With this notation we obtain the identity
	\begin{equation} \label{EnergyEstimate}
	\nabla^\mu J_\mu^X(\psi) = K^X(\psi) \;.
	\end{equation}
	For special choices of $X$ we will integrate this identity over a bounded domain and use the divergence theorem to rewrite the integrated left hand side of \eqref{EnergyEstimate} as a boundary term  to obtain a so-called \emph{energy estimate}.
	
	\subsection{Propagating decay from the horizon to a surface of constant $r$} \label{SecHR}

	In this section we find it convenient to work in the $(v,r,\theta,\varphi)$ coordinates, which were introduced in Section \ref{SecNotation}. It follows from \eqref{Metricvr} that the volume form is given by $\vol = r^2 \sin \theta \, dv \wedge dr \wedge d\theta \wedge d \varphi$.
	The hypersurfaces $\{r = r_0\}$ are spacelike hypersurfaces in the black hole interior. The induced volume form on $\{r = r_0\}$ is denoted by $\volro$, and the future directed normal to $\{r = r_0\}$ by $\nro$. Moreover, we choose $\nh = \frac{\partial}{\partial v}\Big|_r $ as a future directed normal for the horizon $\Hp$ and $\volh = r^2 \sin \theta\, dv \wedge d\theta \wedge d\varphi$ as a volume form. On the hypersurfaces $\{v = v_0\}$ we set $n_{\{v = v_0\}} = -\frac{\partial}{ \partial r}\big|_v$ and $\vol_{\{v = v_0\}} = r^2 \sin \theta\, dr \wedge d\theta \wedge d\varphi$. We prove the following

	\begin{theorem} \label{StabHorR}
		Let $\psi$ be a smooth solution to the  wave equation on the maximal analytic Schwarzschild spacetime that satisfies
		\begin{align} 
		|\psi| &\lesssim v^{-q}   \label{DecayAssumptionsHor1}\\ 
		| \partial_t^{(i)} \Omega_k^{(j)} \psi | &\lesssim v^{-(q + \delta)} \quad \textnormal{ for } i,j \in \N_0, 1 \leq i + j \leq 6, k = 1,2,3 \label{DecayAssumptionsHor2}
		\end{align}
		along the event horizon $\Hp \cap \{v \geq 1\}$, where $q >0$, $\delta \geq 0$.
		
		Then there exists an $r_0 \in (0, 2m)$ close to $2m$ such that 
		\begin{equation}\label{FirstStatementThm}
		|\psi| \lesssim v^{-q} \quad \textnormal{ holds in } \{v \geq 1\} \cap \{r_0 \leq r \leq 2m\}
		\end{equation}
		and for any future directed timelike vector field $N$ on $\{r = r_0\}$ that commutes with $\partial_t$, the following inequality holds true:
		\begin{equation}\label{SecondStatementThm}
		\int\limits_{\{r = r_0\} \cap \{v_0 \leq v \leq v_1\}} J^N\big(\partial_t^{(i)} \Omega_k^{(j)} \psi \big) \cdot \nro \volro \lesssim \big(|v_1- v_0|+1 \big) \cdot v_0^{-2(q + \delta)},
		\end{equation}
		for all $1 \leq v_0 < v_1$, $i,j \in \N_0$, $0 \leq i+j \leq 6$, $k = 1,2,3$.
	\end{theorem}

	\begin{remark}
		Let $N$ be a future directed timelike vector field with $[N, \partial_t] = 0$. It then follows from the assumptions \eqref{DecayAssumptionsHor2} of Theorem \ref{StabHorR} that for $1 \leq v_0 < v_1$ we have
		\begin{equation}\label{AltAssump}
		\int\limits_{\Hp \cap \{v_0 \leq v \leq v_1\} } J^N(\partial_t^{(i)} \Omega_k^{(j)} \psi) \cdot \nh \, \volh \lesssim |v_1 - v_0| \cdot v_0^{-2(q + \delta)}
		\end{equation}
		for all $i,j \in \N_0$, $0 \leq i+j \leq 4$, $k = 1,2,3$. The assumption \eqref{DecayAssumptionsHor2} enters the proof in the form of \eqref{AltAssump}. Hence, Theorem \ref{StabHorR} also holds if we replace \eqref{DecayAssumptionsHor2} by \eqref{AltAssump}.
	\end{remark}

	For the proof of  Theorem \ref{StabHorR} we follow an idea by Luk, \cite{Luk10}. It relies in particular on the redshift vector field of Dafermos and Rodnianski, \cite{DafRod09a}, \cite{DafRod08}, and on the following lemma, the proof of which can be found in \cite{Luk10}, Section 6, and \cite{Fra14}, Section 5, Lemma 5.3. 
	
	\begin{lemma} \label{LemDecay}
		Let $f : [1, \infty) \to \R^+$ satisfy
		\begin{equation*}
		f(v_1) + b \int\limits_{v_0}^{v_1} f(v) \, dv \leq f(v_0) + C\big(|v_1 - v_0|+1\big) \cdot v_0^{-p}
		\end{equation*}
		for all $1 \leq v_0 < v_1$, where $b, C, p$ are positive constants. It then follows that $f(v) \lesssim v^{-p}$.
	\end{lemma}

	\begin{proof}[Proof of Theorem \ref{StabHorR}:]
		By Proposition 3.3.1 of \cite{DafRod08} there exists an $r_0 \in (0, 2m)$ close to $2m$ and a future directed timelike vector field $N$ with $[N, \partial_t] = 0$ such that 
		\begin{equation}\label{RedShift}
		K^N(\psi) \geq b J^N_\mu(\psi) \cdot N^\mu
		\end{equation}
		holds in $\{r_0 \leq r \leq 2m\} \cap \{v \geq 1\}$ for some constant $b>0$.
		We now integrate \eqref{EnergyEstimate} with $X = N$ over the region $\{r_0 \leq r \leq 2m\} \cap \{v_0 \leq v \leq v_1\}$ to obtain
		\begin{equation}\label{Eq0}
		\begin{split}
		\int\limits_{\{r = r_0\} \cap \{v_0 \leq v \leq v_1\}} &J^N(\psi) \cdot \nro \, \volro + \int\limits_{\{r_0 \leq r \leq 2m\} \cap \{v = v_1\}} J^N(\psi) \cdot n_{\{v = v_1\}} \, \vol_{\{v =v_1\}}\\
		&\qquad + \int\limits_{\{r_0 \leq r \leq 2m\} \cap \{v_0 \leq v \leq v_1\}} K^N \, \vol\\
		&=\int\limits_{\{r_0 \leq r \leq 2m\} \cap \{v = v_0\}} J^N(\psi) \cdot n_{\{v = v_0\}} \, \vol_{\{v =v_0\}} + \int\limits_{\Hp \cap \{v_0 \leq v \leq v_1\}} J^N(\psi) \cdot n_{\Hp} \, \vol_{\Hp} \;.
		\end{split}
		\end{equation}
		Together with \eqref{RedShift}  it follows that
		\begin{equation}\label{Eq1}
		\begin{split}
		\int\limits_{\{r_0 \leq r \leq 2m\} \cap \{v = v_1\}} &J^N(\psi) \cdot n_{\{v = v_1\}} \, \vol_{\{v =v_1\}} + b'\int\limits_{v_0}^{v_1} \;\int\limits_{\{r_0 \leq r \leq 2m\} \cap \{v = v'\}} J^N(\psi) \cdot n_{\{v = v'\}} \, \vol_{\{v =v'\}} \, dv'\\
		&\leq\int\limits_{\{r_0 \leq r \leq 2m\} \cap \{v = v_0\}} J^N(\psi) \cdot n_{\{v = v_0\}} \, \vol_{\{v =v_0\}} + \int\limits_{\Hp \cap \{v_0 \leq v \leq v_1\}} J^N(\psi) \cdot n_{\Hp} \, \vol_{\Hp} \;,
		\end{split}
		\end{equation}
		where $b'>0$. Setting $$f(v') = \int\limits_{\{r_0 \leq r \leq 2m\} \cap \{v = v'\}} J^N(\psi) \cdot n_{\{v = v'\}} \, \vol_{\{v =v'\}}$$ and using \eqref{AltAssump} with $i = j =0$ for the last term on the right hand side of \eqref{Eq1}, we obtain
		\begin{equation*}
		f(v_1) + b' \int\limits_{v_0}^{v_1} f(v) \, dv \leq f(v_0) + C\big(|v_1 - v_0|+1\big) \cdot v_0^{-2(q+ \delta)}
		\end{equation*}
		for all $1 \leq v_0 < v_1$. It now follows from Lemma \ref{LemDecay} that $f(v) \lesssim v^{-2(q + \delta)}$. Using this decay for the first term on the right hand side of \eqref{Eq0}, \eqref{AltAssump} for the second term on the right hand side, and dropping the positive second and third term on the left hand side, we obtain
		\begin{equation*}
		\int\limits_{\{r = r_0\} \cap \{v_0 \leq v \leq v_1\}} J^N(\psi) \cdot \nro \, \volro \lesssim \big(|v_1 - v_0| + 1\big) \cdot v_0^{-2(q + \delta)} 
		\end{equation*}
		for all $1 \leq v_0 < v_1$. Note that the vector fields $\partial_t$ and $\Omega_k$ are Killing and thus commute with the wave equation. The same argument then applied to $\partial_t^{(i)} \Omega_k^{(j)} \psi$, $ 0 \leq i + j \leq 6$, $k = 1,2,3$, instead of $\psi$, gives \eqref{SecondStatementThm}.
		
		It remains to prove \eqref{FirstStatementThm}. We note that after commutation with $\Omega_k$, we do not only obtain $f(v) \lesssim v^{-2(q + \delta)}$, but also
		\begin{equation*}
		\sum_{\substack{k = 1,2,3 \\ 0 \leq  j \leq 2}} \;\;\int\limits_{\{r_0 \leq r \leq 2m\} \cap \{v = v'\}} J^N(\Omega_k^{(j)}\psi) \cdot n_{\{v = v'\}} \, \vol_{\{v =v'\}} \lesssim (v')^{-2(q + \delta)} \;.
		\end{equation*}
		Using $\vol_{\{v = v'\}} = r^2 \sin \theta\, dr \wedge d\theta \wedge d\varphi$ this in particular implies for all $r_0 \leq r' \leq 2m$
		\begin{equation*}
		\sum_{\substack{ k = 1,2,3 \\ j = 0,1,2}} \int\limits_{r'}^{2m} \int\limits_{\mathbb{S}^2}\big( \Omega_k^{(j)} \partial_r \psi\big)^2 \Big|_{v = v'}\sin \theta d\theta d\varphi\, r^2 dr \lesssim  (v')^{-2(q + \delta)} \;.
		\end{equation*}
		Sobolev embedding on the spheres gives
		\begin{equation*}
		\int\limits_{r'}^{2m} ||\partial_r \psi ||_{L^\infty(\mathbb{S}^2)}^2\Big|_{v = v', r} dr \lesssim  (v')^{-2(q + \delta)} \;.
		\end{equation*}
		It now follows from \eqref{DecayAssumptionsHor1} that for $v' \geq 1$ and $r' \in [r_0, 2m]$ we have
		\begin{equation*}
		\begin{split}
		|\psi(v',r', \theta, \varphi)| &\leq \int\limits_{r'}^{2m}|\partial_r \psi|(v', r, \theta, \varphi) \, dr + |\psi(v', 2m, \theta, \varphi)| \\
		&\leq (2m - r')^\frac{1}{2} \cdot \Big(\int\limits_{r'}^{2m} ||\partial_r \psi ||_{L^\infty(\mathbb{S}^2)}^2\Big|_{v = v', r} dr\Big)^{\frac{1}{2}} + |\psi(v', 2m, \theta, \varphi)|  \\
		&\lesssim  (v')^{-(q + \delta)} + (v')^{-q} \\
		&\lesssim (v')^{-q} \;.
		\end{split}
		\end{equation*}
		This concludes the proof.
	\end{proof}

	\subsection{Energy estimates from $r=r_0$ to $r = 0$} \label{SecRS}
	
	Let us consider the orthonormal frame adapted to the constant $r$ hypersurfaces $\Sigma_r$:
	\begin{align}\label{adframe}
	e_0=-(\frac{2m}{r}-1)^\frac{1}{2}\partial_r,&&e_1=(\frac{2m}{r}-1)^{-\frac{1}{2}}\partial_t,&&e_2=\frac{1}{r}\partial_\theta,&&e_3=\frac{1}{r\sin\theta}\partial_\varphi.
	\end{align}
	The $e_0$ current of $\psi$ reads
	\begin{align}\label{adJe0}
	J^{e_0}_a(\psi)=(e_0)^b(\partial_a\psi\partial_b\psi-\frac{1}{2}g_{ab}\partial^k\psi\partial_k\psi),&&a,b=0,1,2,3.
	\end{align}
	Note that $J^{e_0}_0(\psi)=\frac{1}{2}[(e_0\psi)^2+|\overline{\nabla}\psi|^2]$, where $\overline{\nabla}$ is the connection intrinsic to $\Sigma_r$. We also use the notation $\slashed{\nabla}$ to denote the connection intrinsic to the round spheres $r\mathbb{S}^2$ of radius $r$.
	The non-vanishing components of the deformation tensor $\pi^{e_0}_{ab}$ of $e_0$ equal
	\begin{align}\label{pie0}
	\pi^{e_0}_{11}=\frac{m}{r^2}(\frac{2m}{r}-1)^{-\frac{1}{2}},\qquad \pi^{e_0}_{22}=\pi^{e_0}_{33}=-(\frac{2m}{r}-1)^\frac{1}{2}\frac{1}{r}
	\end{align}
	Let $\mathcal{D}(\Sigma_{r_0,t_0,T})$ denote the domain of dependence of $\Sigma_{r_0,t_0,T}:=\{r=r_0\}\cap\{t_0\leq t\leq t_0+T\}$, $T>0$, as depicted below:
	\begin{figure}[h]
		\centering
		\def\svgwidth{7cm}
		\includegraphics[scale=1.5]{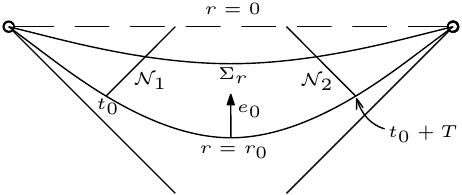}
		\caption{} \label{FigintdomSch}
	\end{figure}

	\noindent
In the following we will use energy estimates in the region $\mathcal{D}(\Sigma_{r_0,t_0,T}) \cap J^-(\Sigma_r)$, $0 <r\leq r_0$, with multipliers of the form $f(r)e_0$, where the weight $f(r)$ is chosen suitably. These kind of $r$-weighted multipliers are adapted to the singular geometry of the Schwarzschild black hole at $r=0$ and are thus very different from the multipliers used in the interior of the Kerr and Reissner-Nordstr\"om black holes, cf.\ \cite{Fra14}, \cite{Gajic15a}, \cite{Gajic15b}, \cite{LukOh15}, \cite{LukSb}.
	
	The affine null generators on the ingoing null hypersurfaces $\mathcal{N}_1,\mathcal{N}_2$ are $n_1=e_0+e_1,n_2=e_0-e_1$ respectively. The intrinsic volume form of $\Sigma_r$ equals
	\begin{align}\label{volSigmarlapse}
	\mathrm{vol}_{\Sigma_r}=(\frac{2m}{r}-1)^\frac{1}{2}r^2dtd\mathbb{S}^2
	\end{align} 
	and it is related to the spacetime volume form via
	\begin{align}\label{volsptim}
	\mathrm{vol}=(\frac{2m}{r}-1)^{-\frac{1}{2}}dr\wedge\mathrm{vol}_{\Sigma_r}.
	\end{align}
	We will use the conclusion of Theorem \ref{StabHorR} as our input. Recall that for fixed $r\in(0,2m)$, the coordinate $v$ is a translation of $t$. 
	
	\subsubsection{Logarithmic upper bound for $\psi$}\label{subsubsec:loguppbd}
	
	We begin by establishing a logarithmic upper bound for the wave.
	
	\begin{proposition}\label{r0to0prop}
		Let $\psi$ be a smooth solution to the wave equation, $\square_g\psi=0$, in the Schwarzschild interior region $\{0< r\leq r_0\}$, $r_0\in(0,2m)$, satisfying $|\psi(r_0,t,\omega)|\lesssim (|t|+1)^{-q}$ and
		\begin{align}\label{initpsir0}
		\int_{\Sigma_{r_0}\cap D(\Sigma_{r_0,t_0,T})}J^{e_0}_0(\partial_t^{(i)}\Omega_k^{(j)}\psi)\mathrm{vol}_{\Sigma_r}\lesssim (T+1)(|t_0|+1)^{-2(q+\delta)},
		\end{align}
		for all $t_0$, $i+j\leq 2$, $k=1,2,3$ and fixed $q,T>0$, $\delta\ge0$. Then, the following weighted energy is bounded by its initial value:
		\begin{align}\label{psieneste0to0}
		r^\frac{3}{2}\int_{\Sigma_r\cap \mathcal{D}(\Sigma_{r_0,t_0,T})}J^{e_0}_0(\partial_t^{(i)}\Omega_k^{(j)}\psi)\mathrm{vol}_{\Sigma_r}\lesssim
		r_0^\frac{3}{2}\int_{\Sigma_{r_0}\cap D(\Sigma_{r_0,t_0,T})}J^{e_0}_0(\partial_t^{(i)}\Omega_k^{(j)}\psi)\mathrm{vol}_{\Sigma_r},
		\end{align}
		for every $r\in(0,r_0]$, $i+j\leq2$. In addition, $\psi$ satisfies the pointwise bound 
		\begin{align}\label{Linftypsiest}
		|\psi(r,t,\omega)|\lesssim |\log r|(|t|+1)^{-q-\delta}+ (|t|+1)^{-q},
		\end{align}
		for all $\{0<r\leq r_0\}$.
	\end{proposition}
	\begin{proof}
		We apply the energy estimate in the spacetime domain $\mathcal{D}(\Sigma_{r_0,t_0,T})\cap J^{-}(\Sigma_r)$, using the vector field $r^\frac{3}{2}e_0$ as a multiplier\footnote{The exponent $\frac{3}{2}$ is motivated as follows: the standard energy identity for the wave equation \emph{with $e_0$ as a multiplier} gives rise to a bulk term $(e_0\partial_t^{(i)}\Omega_k^{(j)}\psi)^2$ the coefficient of which has leading order behaviour $-(\frac{2m}{r}-1)^{-\frac{1}{2}}\text{tr}\pi^{e_0}\sim\frac{3}{2}\frac{1}{r}$, as $r\rightarrow0$. This is non-integrable. Thus, the favourable weight $r^\frac{3}{2}$ helps us absorb this dangerous term and obtain a standard Gronwall type energy estimate.}:
		\begin{align}\label{Stokespsi}
		\notag&r^\frac{3}{2}\int_{\Sigma_r\cap \mathcal{D}(\Sigma_{r_0,t_0,T})}J^{e_0}_0(\psi)\mathrm{vol}_{\Sigma_r}
		+\sum_{l=1,2}\int_{\mathcal{N}_l\cap J^{-}(\Sigma_r)}r^\frac{3}{2}J^{e_0}(\psi)\cdot n_l\mathrm{vol}_{\mathcal{N}_l}+\int_{\mathcal{D}(\Sigma_{r_0,t_0,T})\cap J^{-}(\Sigma_r)}K^{r^\frac{3}{2}e_0}(\psi)\mathrm{vol}\\
		=&\,r_0^\frac{3}{2}\int_{\Sigma_{r_0}\cap D(\Sigma_{r_0,t_0,T})}J^{e_0}_0(\psi)\mathrm{vol}_{\Sigma_r}
		\end{align}
		According to (\ref{pie0}), the bulk integrand equals:
		\begin{align}\label{Kr3/2e0}
		K^{r^\frac{3}{2}e_0}(\psi)=&\,\pi^{r^\frac{3}{2}e_0}_{\mu\nu}T^{\mu\nu}(\psi)=r^\frac{3}{2}\pi^{e_0}_{\mu\nu}T^{\mu\nu}(\psi)-(e_0r^\frac{3}{2})T_{00}(\psi)\\
		\notag=&\,r^\frac{3}{2}
		\frac{m}{r^2}(\frac{2m}{r}-1)^{-\frac{1}{2}}\frac{1}{2}\big[(e_0\psi)^2+(e_1\psi)^2-|\slashed{\nabla} \psi|^2\big]
		-r^\frac{3}{2}\frac{1}{r}(\frac{2m}{r}-1)^\frac{1}{2}\big[(e_0\psi)^2-(e_1\psi)^2\big]\\
		\notag&+(\frac{2m}{r}-1)^\frac{1}{2}\frac{3}{2}\frac{1}{r}r^\frac{3}{2}\frac{1}{2}\big[(e_0\psi)^2+|\overline{\nabla}\psi|^2\big]
		\\
		\notag=&\,(\frac{2m}{r}-1)^\frac{1}{2}r^\frac{3}{2}  \frac{1}{2}\frac{1}{r}\bigg[\frac{m}{r}(\frac{2m}{r}-1)^{-1}\big[(e_0\psi)^2+(e_1\psi)^2-|\slashed{\nabla}\psi|^2\big]\\
		\notag&-2(e_0\psi)^2+2(e_1\psi)^2
		+\frac{3}{2}(e_0\psi)^2+\frac{3}{2}|\overline{\nabla}\psi|^2
		\bigg]\\
		\notag=&\,O(1)(\frac{2m}{r}-1)^\frac{1}{2} r^\frac{3}{2} J^{e_0}_0(\psi)
		+(\frac{2m}{r}-1)^\frac{1}{2}r^\frac{3}{2}  \frac{1}{2}\frac{1}{r}\big[ 4(e_1\psi)^2+|\slashed{\nabla}\psi|^2\big],
		\end{align}
		where in the last line we used the identity $\frac{m}{r}(\frac{2m}{r}-1)^{-1}=\frac{1}{2}+O(r)$.
		We notice that the last term has a positive sign. Dropping this term in (\ref{Stokespsi}), along with the flux terms through $\mathcal{N}_l\cap J^{-}(\Sigma_r)$, yields:
		\begin{align}\label{Stokespsi2}
		r^\frac{3}{2}\int_{\Sigma_r\cap \mathcal{D}(\Sigma_{r_0,t_0,T})}  J^{e_0}_0(\psi)\mathrm{vol}_{\Sigma_r}
		\leq&  \,r^\frac{3}{2}_0\int_{\Sigma_{r_0}\cap D(\Sigma_{r_0,t_0,T})}  J^{e_0}_0(\psi)\mathrm{vol}_{\Sigma_{r_0}}\\
		\notag&+\int^{r_0}_rO(1)s^\frac{3}{2}\int_{\Sigma_s\cap\mathcal{D}(\Sigma_{r_0,t_0,T})}  J^{e_0}_a(\psi) \mathrm{vol}_{\Sigma_s} ds
		\end{align}
		Thus, by Gronwall's inequality we obtain the energy estimate
		\begin{align}\label{psienest}
		r^\frac{3}{2}\int_{\Sigma_r\cap \mathcal{D}(\Sigma_{r_0,t_0,T})}  J^{e_0}_0(\psi)\mathrm{vol}_{\Sigma_r}\lesssim r^\frac{3}{2}_0\int_{\Sigma_{r_0}\cap D(\Sigma_{r_0,t_0,T})}  J^{e_0}_0(\psi)\mathrm{vol}_{\Sigma_{r_0}},
		\end{align}
		for all $r\in(0,r_0]$.
		After commuting with the Killing vector fields $\partial_t^{(i)}\Omega_k^{(j)}$, we arrive at (\ref{psieneste0to0}). In particular, taking into account the volume form (\ref{volSigmarlapse}) and the scaling of $e_0\sim-r^{-\frac{1}{2}}\partial_r$, we have the bound
		\begin{align}\label{e0psienest}
		\sum_{i+j\leq2}\int_{\Sigma_r\cap \mathcal{D}(\Sigma_{r_0,t_0,T})}  r^2(\partial_r\partial_t^{(i)}\Omega_k^{(j)}\psi)^2dtd\mathbb{S}^2\lesssim \sum_{i+j\leq2}\int_{\Sigma_{r_0}\cap \mathcal{D}(\Sigma_{r_0,t_0,T})}  r_0^2(\partial_r\partial_t^{(i)}\Omega_k^{(j)}\psi)^2dtd\mathbb{S}^2
		\end{align}
		for all $r\in(0,r_0]$, $k=1,2,3$.
		
		The pointwise bound (\ref{Linftypsiest}) follows now by first using the fundamental theorem of calculus along the $e_0$ curves and then the Sobolev embedding $H^2([t_0,t_0+T]\times\mathbb{S}^2)\hookrightarrow L^\infty([t_0,t_0+T]\times\mathbb{S}^2)$:
		\begin{align}\label{Linftyest}
		\psi(r,t,\omega)=&\int^{r_0}_{r}\partial_s\psi(s,t,\omega) ds +\psi(r_0,t,\omega)\\
		\tag{C-S}|\psi(r,t,\omega)|\leq&\, \int^{r_0}_{r}|\partial_s\psi(s,t,\omega)| ds +|\psi(r_0,t,\omega)|\\
		\notag\lesssim&\int^{r_0}_{r}\|\partial_s\psi(s,t,\omega)\|_{H^2([t,t+T]\times\mathbb{S}^2)}ds+(|t|+1)^{-q}\\
		\tag{by (\ref{e0psienest}) and initial assumption (\ref{initpsir0})}\lesssim &\,\bigg|\int^{r}_{r_0}\frac{1}{s}ds\bigg|(|t|+1)^{-(q+\delta)}+|(|t|+1)^{-q}\\
		\notag\lesssim&\,|\log \frac{r}{r_0}|\cdot (|t|+1)^{-(q+\delta)}+(|t|+1)^{-q}
		\end{align}
	\end{proof}
	\begin{remark}
		The pointwise estimate (\ref{Linftyest}) uses up to three derivatives of $\psi$ in $L^2$. A corresponding logarithmic bound for the $L^2$ norm of $\psi$ actually holds without using Sobolev embedding. Indeed, taking the $L^2$ norms of both sides in the first line of (\ref{Linftyest}) we have:
		\begin{align}\label{L2psiest}
		\|\psi(t,r,\omega)\|_{L^2([t_0,t_0+T]\times\mathbb{S}^2)}\leq&\int^{r_0}_{r}\|\partial_s\psi\| _{L^2([t_0,t_0+T]\times\mathbb{S}^2)}ds +\|\psi(t,r_0,\omega)\|_{L^2([t_0,t_0+T]\times\mathbb{S}^2)}\\
		\tag{by (\ref{e0psienest}) and initial assumptions}\lesssim &\,\bigg|\int^{r}_{r_0}\frac{1}{s}ds\bigg|(1+T)(|t_0|+1)^{-(q+\delta)}+
		(1+T)(|t_0|+1)^{-q}\\
		\notag\lesssim&\,|\log \frac{r}{r_0}|\cdot (1+T)(|t_0|+1)^{-(q+\delta)}+(1+T)(|t_0|+1)^{-q}
		\end{align}
	\end{remark}
	\begin{remark}
		It is easy to check that the logarithmic bound (\ref{Linftypsiest}) is saturated by locally homogeneous solutions $\psi\big|_U=\psi\big|_U(r)$, cf.\ Remark \ref{rem:symexp}.
	\end{remark}
	\subsubsection{Renormalized energy estimates from $r=r_0$ to $r=0$}

	According to Proposition \ref{r0to0prop}, we should expect a logarithmic leading order behaviour of $\psi$ in the interior region. To prove this rigorously, we derive corresponding energy estimates for the renormalised function $\frac{\psi}{\log r}$ and proceed to prove the asymptotic expansion in Theorem \ref{ThmExp} for $\psi$.
	\begin{theorem}\label{stabr0to0thm}
		Let $\psi$ be a smooth solution to the wave equation, $\square_g\psi=0$, in the Schwarzschild interior region $\{0< r\leq r_0<2m\}$, satisfying $|\psi(r_0,t,\omega)|\lesssim (|t|+1)^{-q}$ and
		\begin{align}\label{initrenpsir0}
		\int_{\Sigma_{r_0}\cap D(\Sigma_{r_0,t_0,T})}J^{e_0}_0(\partial_t^{(i)}\Omega_k^{(j)}\frac{\psi}{\log r})\mathrm{vol}_{\Sigma_{r_0}}\lesssim  (1+T)(|t_0|+1)^{-2q-2\delta}
		\end{align}
		for all $t_0$, $i+j\leq 5$, $k=1,2,3$, and fixed $T,q>0$, $\delta\ge0$. Then 
		the following weighted energy estimate is valid:
		\begin{align}\label{psirenenest}
		&r^\frac{3}{2}|\log r|^4\int_{\Sigma_r\cap D(\Sigma_{r_0,t_0,T})} \big[(e_0\partial_t^{(i)}\Omega_k^{(j)}\frac{\psi}{\log r})^2+|\overline{\nabla}\partial_t^{(i)}\Omega_k^{(j)}\frac{\psi}{\log r}|^2\big]\mathrm{vol}_{\Sigma_r}\\
		\notag  \lesssim&\, r^\frac{3}{2}_0|\log r_0|^4\int_{\Sigma_{r_0}\cap D(\Sigma_{r_0,t_0,T})} J^{e_0}_0(\partial_t^{(i)}\Omega_k^{(j)}\frac{\psi}{\log r})\mathrm{vol}_{\Sigma_{r_0}}\\
		\notag&+\frac{|\log r_0|^3}{r^\frac{1}{2}_0}\int_{\Sigma_{r_0}\cap D(\Sigma_{r_0,t_0,T})} (\partial_t^{(i)}\Omega_k^{(j)}\frac{\psi}{\log r_0})^2\mathrm{vol}_{\Sigma_{r_0}}
		\end{align}
		for every $r\in(0,r_0]$. Moreover, $\psi$ equals
		\begin{align}\label{psiexp}
		\psi=A(t,\omega)\log r+B(t,\omega)+P(r,t,\omega),
		\end{align}
		where $A,B,P$ are smooth functions satisfying $|A|\lesssim (|t|+1)^{-q-\delta},|B|\lesssim (|t|+1)^{-q},|P(r,t,\omega)|\lesssim r|\log r|(|t|+1)^{-q}$ and
		\begin{align}\label{Penest}
		\int_{\Sigma_r\cap D(\Sigma_{r_0,t_0,T})} \big[(e_0\partial_t^{(i)}\Omega_k^{(j)}P)^2+|\overline{\nabla}\partial_t^{(i)}\Omega_k^{(j)}P|^2\big]\mathrm{vol}_{\Sigma_r}\lesssim r^\frac{1}{2}|\log r|^2(1+T)(|t|_0+1)^{-2q}
		\end{align}
		for $i+j\leq 2$, $r\in(0,r_0]$. 
	\end{theorem}
	\begin{proof}
		Let $\phi:=\frac{\psi}{\log r}$. Then $\phi$ satisfies the equation:
		\begin{align}\label{boxvarphi}
		\square\phi=-2(\frac{2m}{r}-1)^\frac{1}{2}\frac{1}{r\log r}e_0\phi
		-\frac{1}{r^2\log r}\phi
		\end{align}
		We consider the weighted multiplier $X=r^\frac{3}{2}|\log r|^4  e_0$ and compute $K^X(\phi)$ using (\ref{pie0}) and (\ref{boxvarphi}):
		\begin{align}\label{divJhat}
		K^X(\phi)=&-r^\frac{3}{2}|\log r|^4 \bigg[2(\frac{2m}{r}-1)^\frac{1}{2}\frac{1}{r\log r}e_0\phi
		+\frac{1}{r^2\log r}\phi
		\bigg]e_0\phi\\
		\notag&+r^\frac{3}{2}|\log r|^4 \bigg[
		\frac{m}{r^2}(\frac{2m}{r}-1)^{-\frac{1}{2}}\frac{1}{2}\big[(e_0\phi)^2+(e_1\phi)^2-|\slashed{\nabla} \phi|^2\big]\\
		\notag&-\frac{1}{r}(\frac{2m}{r}-1)^\frac{1}{2}\big[(e_0\phi)^2-(e_1\phi)^2\big]
		\bigg]\\
		\notag&+(\frac{2m}{r}-1)^\frac{1}{2}\big[\frac{3}{2}\frac{1}{r}+\frac{4}{r\log r}\big]r^\frac{3}{2}|\log r|^4  J^{e_0}_0(\phi)\\
		\notag=&\,(\frac{2m}{r}-1)^\frac{1}{2}r^\frac{3}{2}(\log r)^4\frac{1}{r}\frac{1}{2} \bigg[-\frac{4}{\log r}(e_0\phi)^2
		-(\frac{2m}{r}-1)^{-\frac{1}{2}}\frac{1}{r\log r}2\phi e_0\phi\\
		\notag&+\frac{m}{r}(\frac{2m}{r}-1)^{-1}\big[(e_0\phi)^2+(e_1\phi)^2-|\slashed{\nabla} \phi|^2\big]
		-2(e_0\phi)^2+2(e_1\phi)^2
		+(\frac{3}{2}+\frac{4}{\log r})  \big[(e_0\phi)^2+|\overline{\nabla} \phi|^2\big]\bigg]\\
		\notag=&\,O(1)(\frac{2m}{r}-1)^\frac{1}{2}r^\frac{3}{2}|\log r|^4  J^{e_0}_0(\phi)
		+(\frac{2m}{r}-1)^\frac{1}{2}r^\frac{3}{2}(\log r)^4\frac{1}{r}\big[
		(2+\frac{2}{\log r})(e_1\phi)^2+(\frac{1}{2}+\frac{2}{\log r})|\slashed{\nabla}\phi|^2\big]\\
		\tag{$\frac{m}{r}(\frac{2m}{r}-1)^{-1}=\frac{1}{2}+O(r)$}&+(\frac{2m}{r}-1)^\frac{1}{2}\frac{(\log r)^3}{r^\frac{1}{2}}\phi \partial_r\phi
		\end{align}
		Then the usual energy estimate with vector field $X=r^\frac{3}{2}|\log r|^4 e_0$, applied to the region $D(\Sigma_{r_0,t_0,T})\cap J^{-}(\Sigma_r)$, gives the following inequality:\footnote{Note that the flux terms through the null boundaries of the region have a favourable sign and therefore can be dropped.}
		\begin{align}\label{enidvarphi}
		\notag&\frac{1}{2}r^\frac{3}{2}|\log r|^4\int_{\Sigma_r\cap D(\Sigma_{r_0,t_0,T})} \big[(e_0\phi)^2+|\overline{\nabla}\phi|^2\big]\mathrm{vol}_{\Sigma_r} \\
		\leq&\,\frac{1}{2}r^\frac{3}{2}_0|\log r_0|^4\int_{\Sigma_{r_0}\cap D(\Sigma_{r_0,t_0,T})} \big[(e_0\phi)^2+|\overline{\nabla}\phi|^2\big]\mathrm{vol}_{\Sigma_{r_0}}\\
		\notag&-\int^{r_0}_r\int_{\Sigma_s\cap D(\Sigma_{r_0,t_0,T})}\frac{1}{s}s^\frac{3}{2}|\log s|^4 \big[(2+\frac{2}{\log s})|e_1\phi|^2+(\frac{1}{2}+\frac{2}{\log s})|\slashed{\nabla}\phi|^2\big]\mathrm{vol}_{\Sigma_s} ds\\
		\notag&-\int^{r_0}_r\int_{\Sigma_s\cap D(\Sigma_{r_0,t_0,T})}\bigg[\frac{(\log s)^3}{s^\frac{1}{2}} \phi\partial_s\phi
		+O(1)s^\frac{3}{2}|\log s|^4  J^{e_0}_0(\phi)\bigg]
		\mathrm{vol}_{\Sigma_s} ds
		\end{align}
Observe that the terms in the third line of (\ref{enidvarphi}), whose coefficients come from the deformation tensor of the multiplier $X=r^\frac{3}{2}|\log r|^4e_0$, are non-integrable in $[0,r_0]$ and potentially uncontrollable. However, we notice that for $s<e^{-4}$
they have a negative sign and therefore the part of the integral from $[0,e^{-4}]$ can be dropped. What remains from these terms, integrated in $[e^{-4},r_0]$ (if $r_0>e^{-4}$), can be thus incorporated in the $O(1)$ term in the end of \eqref{enidvarphi}. Taking this into account, we proceed by integrating by parts the term $\phi\partial_s\phi$ in the bulk:
		\begin{align}\label{enestvarphi}
		&\frac{1}{2}r^\frac{3}{2}|\log r|^4\int_{\Sigma_r\cap D(\Sigma_{r_0,t_0,T})} \big[(e_0\phi)^2+|\overline{\nabla}\phi|^2\big]\mathrm{vol}_{\Sigma_r}-\frac{1}{2}\frac{(\log r)^3}{r^\frac{1}{2}}\int_{\Sigma_r\cap D(\Sigma_{r_0,t_0,T})} \phi^2\mathrm{vol}_{\Sigma_r}\\
		\notag  \leq&\, \frac{1}{2}r^\frac{3}{2}_0|\log r_0|^4\int_{\Sigma_{r_0}\cap D(\Sigma_{r_0,t_0,T})} \big[(e_0\phi)^2+|\overline{\nabla}\phi|^2\big]\mathrm{vol}_{\Sigma_{r_0}}-\frac{1}{2}\frac{(\log r_0)^3}{r^\frac{1}{2}_0}\int_{\Sigma_{r_0}\cap D(\Sigma_{r_0,t_0,T})} \phi^2\mathrm{vol}_{\Sigma_{r_0}} \\
		\tag{$\mathrm{vol}_{\Sigma_s}=[\sqrt{2m}s^\frac{3}{2}+O(s^\frac{5}{2})]dtd\mathbb{S}^2$}&+\int^{r_0}_r\int_{\Sigma_s\cap D(\Sigma_{r_0,t_0,T})}\bigg[\frac{1}{2}\big[\frac{3(\log s)^2}{s^\frac{3}{2}}-\frac{1}{2}\frac{(\log s)^3}{s^\frac{3}{2}}+\frac{3}{2}\frac{(\log s)^3}{s^\frac{3}{2}}\big] \phi^2\\
		\notag&+\frac{(\log s)^3}{s^\frac{1}{2}}O(1) \phi^2+O(1)s^\frac{3}{2}|\log s|^4  J^{e_0}_0(\phi)\bigg]
		\mathrm{vol}_{\Sigma_s} ds
		\end{align}
		Note that $\frac{3(\log s)^2}{s^\frac{3}{2}}-\frac{1}{2}\frac{(\log s)^3}{s^\frac{3}{2}}+\frac{3}{2}\frac{(\log s)^3}{s^\frac{3}{2}}<0$, for $s$ small. 
		Hence, the zeroth order term in the fourth line of (\ref{enestvarphi}) can be dropped and a standard Gronwall's inequality can be applied to obtain the energy estimate:
		\begin{align}\label{enestvarphi2}
		&r^\frac{3}{2}|\log r|^4\int_{\Sigma_r\cap D(\Sigma_{r_0,t_0,T})} \big[(e_0\phi)^2+|\overline{\nabla}\phi|^2\big]\mathrm{vol}_{\Sigma_r}+\frac{|\log r|^3}{r^\frac{1}{2}}\int_{\Sigma_r\cap D(\Sigma_{r_0,t_0,T})} \phi^2\mathrm{vol}_{\Sigma_r}\\
		\notag  \lesssim&\,r^\frac{3}{2}_0|\log r_0|^4\int_{\Sigma_{r_0}\cap D(\Sigma_{r_0,t_0,T})} \big[(e_0\phi)^2+|\overline{\nabla}\phi|^2\big]\mathrm{vol}_{\Sigma_{r_0}}+\frac{|\log r_0|^3}{r^\frac{1}{2}_0}\int_{\Sigma_{r_0}\cap D(\Sigma_{r_0,t_0,T})}\phi^2\mathrm{vol}_{\Sigma_{r_0}},
		\end{align}
		for all $r\in(0,r_0]$. The preceding estimate is also valid for $\partial_t^{(i)}\Omega^{(j)}_k\phi$, since $\partial_t^{(i)}\Omega^{(j)}_k$ commutes with the equation (\ref{boxvarphi}). This completes the proof of (\ref{psirenenest}).
		
		The estimate (\ref{psirenenest}) implies that $\frac{\psi}{\log r}$ has a limit in $L^2([t_0,t_0+T]\times\mathbb{S}^2])$, as $r\rightarrow0$. Indeed, given $r_1,r_2\in(0,r_0]$ it holds:
		\begin{align}\label{psirenLinftyest}
		\|\frac{\psi(r_2,t,\omega)}{\log r_2}-\frac{\psi(r_1,t,\omega)}{\log r_1}\|_{L^2([t_0,t_0+T]\times\mathbb{S}^2])}=\big\|\int^{r_2}_{r_1}\partial_s\frac{\psi}{\log s} ds\big\|_{L^2([t_0,t_0+T]\times\mathbb{S}^2])}\\
		\notag\lesssim \bigg|\int^{r_2}_{r_1}\|\partial_s\frac{\psi}{\log s}\|_{L^2([t_0,t_0+T]\times\mathbb{S}^2])}ds\bigg|\lesssim\bigg|\int^{r_2}_{r_1}\frac{1}{s(\log s)^2}ds\bigg|\lesssim\big|\frac{1}{\log r_2}-\frac{1}{\log r_1}\big|,
		\end{align}
		which implies that the function $\frac{\psi}{\log r}:(0,r_0]\to L^2([t_0,t_0+T]\times\mathbb{S}^2)$ is uniformly continuous and hence it extends continuously to $r=0$. Define 
		\begin{align}\label{Adef}
		A(t,\omega):\overset{L^2}{=}\lim_{r\rightarrow0}\frac{\psi}{\log r} 
		\end{align}
		The smoothness of $A$ follows by repeating the above argument locally for $\partial_t^{(i)}\Omega_k^{(j)}\frac{\psi}{\log r}$, for all $i,j$. 
		Also, we compute:
		\begin{align}\label{psi-Alogrenlim}
		&r^\frac{3}{2}\int_{\Sigma_r\cap D(\Sigma_{r_0,t_0,T})}[e_0(\psi-A\log r)]^2\mathrm{vol}_{\Sigma_r}\\
		\notag=&\,r^\frac{3}{2}\int_{\Sigma_r\cap D(\Sigma_{r_0,t_0,T})}\big[\log r\,e_0\frac{\psi}{\log r}-(\frac{2m}{r}-1)^\frac{1}{2}\frac{1}{r}(\frac{\psi}{\log r}-A)\big]^2\mathrm{vol}_{\Sigma_r}\\
		\notag\leq&\,r^\frac{3}{2}\int_{\Sigma_r\cap D(\Sigma_{r_0,t_0,T})}2|\log r|^2(e_0\frac{\psi}{\log r})^2+2(\frac{2m}{r}-1)\frac{1}{r^2}(\frac{\psi}{\log r}-A)^2\mathrm{vol}_{\Sigma_r}\\
		\tag{by (\ref{psirenenest})}\lesssim&\,\frac{1}{|\log r|^2}+\int_{\Sigma_r\cap D(\Sigma_{r_0,t_0,T})}(\frac{\psi}{\log r}-A)^2dtd\mathbb{S}^2\overset{r\rightarrow0}{\longrightarrow}0
		\end{align}
		On the other hand, it holds
		\begin{align}\label{Alogrenlim}
		r^\frac{3}{2}\int_{\Sigma_r\cap D(\Sigma_{r_0,t_0,T})}[e_0(A\log r)]^2\mathrm{vol}_{\Sigma_r}=&\,r^\frac{3}{2}(\frac{2m}{r}-1)^\frac{3}{2}\int_{\Sigma_r\cap D(\Sigma_{r_0,t_0,T})}A^2(t,\omega)dtd\mathbb{S}^2\\
		\notag\overset{r\rightarrow0}{\longrightarrow}&\,(2m)^\frac{3}{2}\int_{\Sigma_r\cap D(\Sigma_{r_0,t_0,T})}A^2(t,\omega)dtd\mathbb{S}^2
		\end{align}
		and hence, by (\ref{psi-Alogrenlim}) we have
		\begin{align}\label{e0psienlim}
		r^\frac{3}{2}\int_{\Sigma_r\cap D(\Sigma_{r_0,t_0,T})}(e_0\psi)^2\mathrm{vol}_{\Sigma_r}\overset{r\rightarrow0}{\longrightarrow}&\,(2m)^\frac{3}{2}\int_{\Sigma_r\cap D(\Sigma_{r_0,t_0,T})}A^2(t,\omega)dtd\mathbb{S}^2.
		\end{align}
		Iterating the above argument with $\partial_t^{(i)}\Omega_k^{(j)}$ and employing (\ref{initpsir0}),(\ref{psieneste0to0}), we conclude that
		\begin{align}\label{Aenest}
		\int_{\Sigma_r\cap D(\Sigma_{r_0,t_0,T})}(\partial_t^{(i)}\Omega_k^{(j)}A)^2dtd\mathbb{S}^2\lesssim (1+T)(|t_0|+1)^{-2q-2\delta},&&i+j\leq 5.
		\end{align}
		The latter bound also implies that $|A|\lesssim \|A\|_{H^2([t,t+T]\times\mathbb{S}^2)}\lesssim t^{-q-\delta}$. 
		
		Let $\psi_1(r,t,\omega):=\psi(r,t,\omega)-A(t,\omega)\log r$. Then $\psi_1$ satisfies
		\begin{align}\label{boxpsi1}
		\square\psi_1=-\Delta_{\mathbb{S}^2}A\frac{\log r}{r^2}-\frac{A}{r^2}-(\frac{2m}{r}-1)^{-1}\partial_t^2A
		\end{align}
		According to (\ref{psi-Alogrenlim}), the $e_0$ part of the integral of the weighted energy current $J^{r^\frac{3}{2}e_0}(\partial_t^{(i)}\Omega_k^{(j)}\psi_1)$, $i+j\leq5$, tends to zero. The same can be easily seen to hold true for the spatial part of the preceding energy:
		\begin{align}\label{psi1enlim}
		r^\frac{3}{2}\int_{\Sigma_r\cap D(\Sigma_{r_0,t_0,T})}  |\overline{\nabla}\partial_t^{(i)}\Omega_k^{(j)}\psi_1|^2\mathrm{vol}_{\Sigma_r}
		=&\,r^\frac{3}{2}|\log r|^2\int_{\Sigma_r\cap D(\Sigma_{r_0,t_0,T})} |\overline{\nabla}\partial_t^{(i)}\Omega_k^{(j)}(\frac{\psi}{\log r}-A)|^2\mathrm{vol}_{\Sigma_r}\\
		\tag{by (\ref{psirenenest})}\lesssim&\,\frac{1}{|\log r|^2}+r^3|\log r|^2\int_{\Sigma_r\cap D(\Sigma_{r_0,t_0,T})}|\overline{\nabla}\partial_t^{(i)}\Omega_k^{(j)} A|^2dtd\mathbb{S}^2\\
		\tag{$e_1\sim\sqrt{r}\partial_t,\,e_{2,3}\sim\frac{1}{r}\partial_{\theta,\varphi}$}\lesssim&\,\frac{1}{|\log r|^2}+r|\log r|^2\overset{r\rightarrow0}{\longrightarrow}0
		\end{align}
		Hence, applying the energy estimate to $\partial_t^{(i)}\Omega_k^{(j)}\psi_1$, $i+j\leq3$, with the vector $X=r^\frac{3}{2}e_0$ in the domain $D(\Sigma_{r_0,t_0,T})\cap\{0<s\leq r\}$, utilizing (\ref{Kr3/2e0}) for $\partial_t^{(i)}\Omega_k^{(j)}\psi_1$ with the additional term $r^\frac{3}{2}e_0\psi_1\cdot\square\psi_1$ in the RHS, we obtain the following (the null flux terms have an unfavourable sign in this case):
		\begin{align}\label{Stokespsi1}
		&r^\frac{3}{2}\int_{\Sigma_r\cap D(\Sigma_{r_0,t_0,T})} \big[(e_0\partial_t^{(i)}\Omega_k^{(j)}\psi_1)^2+|\overline{\nabla}\partial_t^{(i)}\Omega_k^{(j)}\psi_1|^2\big]\mathrm{vol}_{\Sigma_r}\\
		\tag{this term vanishes by (\ref{psi-Alogrenlim}),(\ref{psi1enlim})}=&\,\lim_{r\rightarrow0}r^\frac{3}{2}\int_{\Sigma_r\cap D(\Sigma_{r_0,t_0,T})} \big[(e_0\partial_t^{(i)}\Omega_k^{(j)}\psi_1)^2+|\overline{\nabla}\partial_t^{(i)}\Omega_k^{(j)}\psi_1|^2\big]\mathrm{vol}_{\Sigma_r}\\
		\tag{recall that $n_{1,2}=e_0\pm e_1$}&+\sum_{l=1,2}\int^r_0\int_{\mathbb{S}^2} s^\frac{3}{2}\big[(n_l\partial_t^{(i)}\Omega_k^{(j)}\psi_1)^2+|\slashed{\nabla}\partial_t^{(i)}\Omega_k^{(j)}\psi_1|^2\big](\frac{2m}{s}-1)^{-\frac{1}{2}} s^2d\mathbb{S}^2ds\\
		\notag&+\int^r_0s^\frac{3}{2}\int_{\Sigma_s\cap D(\Sigma_{r_0,t_0,T})} \big[\frac{4}{s}(e_1\partial_t^{(i)}\Omega_k^{(j)}\psi_1)^2+\frac{1}{s}|\slashed{\nabla}\partial_t^{(i)}\Omega_k^{(j)}\psi_1|^2\big]+O(1)J^{e_0}_0(\partial_t^{(i)}\Omega_k^{(j)}\psi_1)\mathrm{vol}_{\Sigma_s} ds\\
		&\notag-\int^r_0(\frac{2m}{s}-1)^{-\frac{1}{2}}s^\frac{3}{2}\int_{\Sigma_s\cap D(\Sigma_{r_0,t_0,T})}2  e_0\partial_t^{(i)}\Omega_k^{(j)}\psi_1\bigg[\Delta_{\mathbb{S}^2}\partial_t^{(i)}\Omega_k^{(j)}A\frac{\log s}{s^2}+\frac{\partial_t^{(i)}\Omega_k^{(j)}A}{s^2}\\
		\notag&+(\frac{2m}{s}-1)^{-1}\partial_t^2\partial_t^{(i)}\Omega_k^{(j)}A\bigg]\mathrm{vol}_{\Sigma_s} ds
		\\
		\tag{by 1D Sobolev in $t$}\lesssim&\sum_{l=1,2}\int^r_0\int_{\Sigma_s\cap D(\Sigma_{r_0,t_0,T})} s^4\big[(n_l\partial_t^{(i+1)}\Omega_k^{(j)}\psi_1)^2+|\slashed{\nabla}\partial_t^{(i+1)}\Omega_k^{(j)}\psi_1|^2\big]dtd\mathbb{S}^2ds\\
		\notag&+\int^r_0\int_{\Sigma_s\cap D(\Sigma_{r_0,t_0,T})}s^3(\partial_t^{(i+1)}\Omega_k^{(j)}\psi_1)^2+\sum_{l=1}^3(\Omega_l\partial_t^{(i)}\Omega_k^{(j)}\psi_1)^2+s^3J^{e_0}_0(\partial_t^{(i)}\Omega_k^{(j)}\psi_1)dtd\mathbb{S}^2ds\\
		\notag&+\int^r_0|\log s|\|A\|_{H^5([t_0,t_0+T]\times\mathbb{S}^2)}\bigg(s^\frac{3}{2}\int_{\Sigma_s\cap D(\Sigma_{r_0,t_0,T})} (e_0\partial_t^{(i)}\Omega_k^{(j)}\psi_1)^2\mathrm{vol}_{\Sigma_s}\bigg)^\frac{1}{2}ds\\
		\tag{by (\ref{psi-Alogrenlim}),(\ref{psi1enlim}) and the $L^2$ estimate (\ref{L2psiest}) for $\partial_t^{(i')}\Omega_k^{(j')}\psi_1$, $i'+j'\leq4$}\lesssim&\,r|\log r|^2(1+T)(|t_0|+1)^{-2q}\\
		\notag&+\int^r_0|\log s|\sqrt{1+T}(|t_0|+1)^{-q-\delta}\bigg(s^\frac{3}{2}\int_{\Sigma_s\cap D(\Sigma_{r_0,t_0,T})} (e_0\partial_t^{(i)}\Omega_k^{(j)}\psi_1)^2\mathrm{vol}_{\Sigma_s}\bigg)^\frac{1}{2}ds
		\end{align}
		Applying now the following Gronwall type inequality:\footnote{Proof of (\ref{Gron}): Let $F(r):=|h_0|^2+\int^r_0|h_1(s)||f(s)|ds$. Then $\partial_rF=|h_1||f|\leq |h_1|\sqrt{F}$ or $2\partial_r\sqrt{F}\leq |h_1|$, which after integrating yields $2|f|\leq 2\sqrt{F}\leq 2|h_0|+\int^r_0|h_1|ds$. If $|h_0|=|h_0(r)|$ is non-decreasing, which is the case at hand (for small $r$), then the same inequality holds true. Indeed, replacing $|h_0(r)|\leq|h_0(R)|$, $r\leq R$, and applying (\ref{Gron}), we evaluate the resulting estimate at $r=R$. Since $R$ is arbitrary, the conclusion follows.}
		\begin{align}\label{Gron}
		|f(r)|^2\leq |h_0|^2+\int^r_0|h_1(s)||f(s)|ds\qquad\Longrightarrow\qquad2|f(r)|\leq 2|h_0|+\int^r_0|h_1(s)|ds
		\end{align}
		to (\ref{Stokespsi1}) we conclude that 
		\begin{align}\label{psi1enest}
		r^\frac{1}{2}\int_{\Sigma_r} \big[(e_0\partial_t^{(i)}\Omega_k^{(j)}\psi_1)^2+|\overline{\nabla}\partial_t^{(i)}\Omega_k^{(j)}\psi_1|^2\big]\mathrm{vol}_{\Sigma_r}\lesssim |\log r|^2(1+T)(|t_0|+1)^{-2q},
		\end{align}
		for all $r\in(0,r_0]$, $i+j\leq 3$.
		Arguing now analogously to (\ref{Linftyest}) for $\psi_1$, we prove that $\psi_1$ is bounded, it has a smooth limit at $r=0$, which we define $B(t,\omega):\overset{L^2}{=}\lim_{r\rightarrow0}\psi_1$, and the pointwise decay $|\psi_1|,|B|\lesssim (|t|+1)^{-q}$ is valid everywhere.
		
		Let $P:=\psi_1-B=\psi-A\log r-B$. Then by (\ref{psi-Alogrenlim}),(\ref{psi1enlim}) $r^\frac{3}{2}\int_{\Sigma_r}(e_0\partial_t^{(i)}\Omega_k^{(j)}P)^2+|\overline{\nabla}\partial_t^{(i)}\Omega_k^{(j)}P|^2\mathrm{vol}_{\Sigma_r}\rightarrow0$ and by definition $\int_{\Sigma_r\cap D(\Sigma_{r_0,t_0,T})}(\partial_t^{(i)}\Omega_k^{(j)}P)^2dtd\mathbb{S}^2\rightarrow0$, as $r\rightarrow0$, for $i+j\leq3$. Also, using (\ref{psi1enest}) for $\partial_r P=\partial_r\psi_1$ and C-S we deduce that
		\begin{align}\label{P1H4est}
		&\frac{1}{2}\partial_r\int_{\Sigma_r\cap D(\Sigma_{r_0,t_0,T})}(\partial_t^{(i)}\Omega_k^{(j)}P)^2dtd\mathbb{S}^2\\
		\notag\leq&\,\bigg(\int_{\Sigma_r\cap D(\Sigma_{r_0,t_0,T})}(\partial_t^{(i)}\Omega_k^{(j)}P)^2dtd\mathbb{S}^2\bigg)^\frac{1}{2}\bigg(\int_{\Sigma_r\cap D(\Sigma_{r_0,t_0,T})}(\partial_r\partial_t^{(i)}\Omega_k^{(j)}P)^2dtd\mathbb{S}^2\bigg)^\frac{1}{2}\\
		\tag{$i+j\leq 3$}\leq&\,\frac{|\log r|}{\sqrt{r}}(1+T)(|t_0|+1)^{-q}\bigg(\int_{\Sigma_r\cap D(\Sigma_{r_0,t_0,T})}(\partial_t^{(i)}\Omega_k^{(j)}P)^2dtd\mathbb{S}^2\bigg)^\frac{1}{2}
		\end{align}
		Integrating on $[0,r]$ it follows that
		\begin{align}\label{PH4est2}
		\int_{\Sigma_r\cap D(\Sigma_{r_0,t_0,T})}(\partial_t^{(i)}\Omega_k^{(j)}P)^2dtd\mathbb{S}^2\lesssim r|\log r|^2(1+T)(|t_0|+1)^{-2q},&&r\in(0,r_0],\;i+j\leq3. 
		\end{align}
		Furthermore, $P$ satisfies the equation 
		\begin{align}\label{boxu}
		\square P=-\Delta_{\mathbb{S}^2}A\frac{\log r}{r^2}-\frac{A+\Delta_{\mathbb{S}^2}B}{r^2}-(\frac{2m}{r}-1)^{-1}(\partial_t^2A\log r+\partial^2_tB).
		\end{align}
		Arguing as in (\ref{Stokespsi1}), we arrive at the identity [for $i+j\leq2$]:
		\begin{align}\label{StokesP}
		&r^\frac{3}{2}\int_{\Sigma_r\cap D(\Sigma_{r_0,t_0,T})} \big[(e_0\partial_t^{(i)}\Omega_k^{(j)}P)^2+|\overline{\nabla}\partial_t^{(i)}\Omega_k^{(j)}P|^2\big]\mathrm{vol}_{\Sigma_r}\\
		\notag=&\sum_{l=1,2}\int^r_0\int_{\mathbb{S}^2} s^\frac{3}{2}\big[(n_l\partial_t^{(i)}\Omega_k^{(j)}P)^2+|\slashed{\nabla}\partial_t^{(i)}\Omega_k^{(j)}P|^2\big](\frac{2m}{s}-1)^{-\frac{1}{2}} s^2d\mathbb{S}^2ds\\
		\notag&+\int^r_0s^\frac{3}{2}\int_{\Sigma_s\cap D(\Sigma_{r_0,t_0,T})} \big[\frac{4}{s}(e_1\partial_t^{(i)}\Omega_k^{(j)}P)^2+\frac{1}{s}|\slashed{\nabla}\partial_t^{(i)}\Omega_k^{(j)}P|^2\big]+O(1)J^{e_0}_0(\partial_t^{(i)}\Omega_k^{(j)}P)\mathrm{vol}_{\Sigma_s} ds\\
		&\notag-\int^r_0(\frac{2m}{s}-1)^{-\frac{1}{2}}s^\frac{3}{2}\int_{\Sigma_s\cap D(\Sigma_{r_0,t_0,T})}2  e_0\partial_t^{(i)}\Omega_k^{(j)}P\bigg[\Delta_{\mathbb{S}^2}\partial_t^{(i)}\Omega_k^{(j)}A\frac{\log s}{s^2}+\frac{\partial_t^{(i)}\Omega_k^{(j)}A+\Delta_{\mathbb{S}^2}\partial_t^{(i)}\Omega_k^{(j)}B}{s^2}\\
		\notag&+(\frac{2m}{s}-1)^{-1}\big(\partial_t^2\partial_t^{(i)}\Omega_k^{(j)}A\log s+\partial^2_t\partial_t^{(i)}\Omega_k^{(j)}B\big)\bigg]\mathrm{vol}_{\Sigma_s} ds\\
		\notag\lesssim&\sum_{l=1,2}\int^r_0\int_{\mathbb{S}^2} s^4\big[(n_l\partial_t^{(i+1)}\Omega_k^{(j)}P)^2+|\slashed{\nabla}\partial_t^{(i+1)}\Omega_k^{(j)}P|^2\big]dtd\mathbb{S}^2ds\\
		\notag&+\int^r_0\int_{\Sigma_s\cap D(\Sigma_{r_0,t_0,T})} s^3(\partial_t^{(i+1)}\Omega_k^{(j)}P)^2+\sum_{l=1}^3|\Omega_l\partial_t^{(i)}\Omega_k^{(j)}P|^2+s^3J^{e_0}_0(\partial_t^{(i)}\Omega_k^{(j)}P)dtd\mathbb{S}^2ds\\
		\notag&+\int^r_0|\log s|(\|A\|_{H^4([t_0,t_0+T]\times\mathbb{S}^2)}+\|B\|_{H^4([t_0,t_0+T]\times\mathbb{S}^2)})\bigg(s^\frac{3}{2}\int_{\Sigma_s\cap D(\Sigma_{r_0,t_0,T})} (e_0\partial_t^{(i)}\Omega_k^{(j)}P)^2\mathrm{vol}_{\Sigma_s} \bigg)^\frac{1}{2}ds\\
		\tag{applying (\ref{PH4est2}) to $\partial_t^{(i')}\Omega_k^{(j')}P$, $i'+j'\leq3$}
		\lesssim&\, r^2|\log r|^2(1+T)(|t_0|+1)^{-2q}\\
		\notag&+\int^r_0|\log s|\sqrt{1+T}(|t_0|+1)^{-q}\bigg(s^\frac{3}{2}\int_{\Sigma_s\cap D(\Sigma_{r_0,t_0,T})} (e_0\partial_t^{(i)}\Omega_k^{(j)}P)^2\mathrm{vol}_{\Sigma_s} \bigg)^\frac{1}{2}ds,
		\end{align}
		Thus, employing (\ref{Gron}) we obtain the bound:
		\begin{align}\label{Penest2}
		\int_{\Sigma_r\cap D(\Sigma_{r_0,t_0,T})} \big[(e_0\partial_t^{(i)}\Omega_k^{(j)}P)^2+|\overline{\nabla}\partial_t^{(i)}\Omega_k^{(j)}P|^2\big]\mathrm{vol}_{\Sigma_r}\lesssim r^\frac{1}{2}|\log r|^2(1+T)(|t_0|+1)^{-2q},
		\end{align}
		for every $r\in(0,r_0],\;i+j\leq2$. Finally, by definition $P$ has a zero pointwise limit at $r=0$. Hence, integrating $\partial_rP$ in $[0,r]$ we have:
		\begin{align}\label{PLinftyest}
		P(r,t,\omega)=&\int^{r}_0\partial_sP(s,t,\omega) ds \\
		\notag|P(r,t,\omega)|\lesssim&\int^{r}_{0}|\partial_sP(s,t,\omega)| ds\\
		\notag\lesssim&\,\int^r_0\|\partial_sP(s,t,\omega)\|_{H^2([t,t+T]\times\mathbb{S}^2)} ds\\
		\lesssim&\,r|\log r|(|t|+1)^{-q}\tag{by (\ref{Penest2})}
		\end{align}
		This completes the proof of the theorem.
	\end{proof}
	\subsection{Concluding remark on the continuous dependence of $A(t,\omega)$}

	We conclude this section with the following corollary, which is needed for the next section. It emphasizes the continuous dependency of $A(t,\omega)$ on initial data, given on a spacelike hypersurface $\Sigma'$, as in Theorem \ref{ThmSphSym} a).
	\begin{corollary} \label{CorContDepA}
		Let $\psi$ be a smooth solution to the wave equation on the maximal analytic Schwarzschild spacetime and consider a hypersurface $\Sigma'$ as in Theorem \ref{ThmSphSym} a).
		We denote the future normal of $\Sigma' \cap \{0 < r < 2m\}$ with $n_{\Sigma' \cap \{0 < r < 2m\}}$ and the induced volume form with $\vol_{\Sigma' \cap \{0 < r < 2m\}}$.
		Let $N$ be a future directed timelike vector field, invariant under the flow of $\partial_t$, and let $\mu >0$. Assume that
		\begin{equation}\label{AssumptionCorContDep}
		\begin{aligned}
		\int\limits_{\Sigma' \cap \{0 < r < 2m\}} J^N(\partial_t^{(i)} \Omega_k^{(j)} \psi) \cdot n_{\Sigma' \cap \{0 < r < 2m\}} \, \vol_{\Sigma' \cap \{0 < r < 2m\}} &\leq \mu \\
		\int\limits_{\Sigma' \cap \Hp \cap\{v_1 \leq v \leq v_2\}} J^N(\partial_t^{(i)} \Omega_k^{(j)} \psi) \cdot n_{\Hp} \, \vol_{\Hp} &\leq \mu |v_2 - v_1| \cdot v_1^{-2(q + \delta)}
		\end{aligned}
		\end{equation}
		holds for $i,j \in \N_0$, $0 \leq i+j \leq 4$, $k = 1,2,3$, and $v_0 \leq v_1 < v_2$, together with the analogous estimate for the portion of $\Sigma'$ on the left event horizon. Then $\psi$  satisfies the asymptotic expansion \eqref{AsymptoticExp} near the singularity with
		\begin{equation}\label{ConclusionCorContDep}
		|A(t,\omega)| \leq C(\mu) (|t| + 1)^{-(q + \delta)} \;,
		\end{equation}
		where $C(\mu) \to 0$ for $\mu \to 0$.
	\end{corollary}
	
	\begin{proof}
		The proof is a trivial modification of the results obtained in this section, the difference being that here we start from a hypersurface of the form $\Sigma'$  and we keep track of the dependence of the constants in the energy estimates on the initial data. 
	\end{proof}
	
	\section{An open set of waves blowing up on all of $\{r = 0\}$: proof of Theorem \ref{ThmOpenSet}} \label{SecOpenSet}

	The proof of Theorem \ref{ThmOpenSet} proceeds by first constructing a spherically symmetric solution which blows up on all of $\{r = 0\}$. Here, Theorem \ref{ThmSphSym} a) is used. The open set of waves blowing up on all of $\{r = 0\}$ is then constructed by showing that sufficiently small perturbations of this spherically symmetric solution still blow up on all of $\{r = 0\}$. Here, Corollary \ref{CorContDepA} is used.

	Given a Cauchy hypersurface $\Sigma$ as in Figure \ref{FigSchw} we denote the future directed normal by $n_\Sigma$. We introduce an energy space $\mathcal{E}$ of smooth initial data $(\overline{\psi}_0, \overline{\psi}_1) = (\psi|_\Sigma, n_\Sigma \psi)$ on $\Sigma$ such that an energy norm $E$ of the following form is finite (see \cite{AnArGa18} for details)
	\begin{equation}\label{EnergyNorm}
	\begin{split}
	E(\psi|_\Sigma, n_\Sigma \psi) &= \sum_{i,j,k, \Gamma} \int\limits_{\Sigma \cap \{r \geq R\}}  \Big(r^{i_0} | \slashed\nabla^{i_1}\partial_v^{i_2}(r^{i_3}\partial_v^{i_4}\Gamma \phi)|^2 \\
	&\qquad+ r^{j_0}|\slashed\nabla^{j_1}\partial_v^{j_2}(r^2\partial_v)^{j_3} \Gamma \phi|^2 + r^{k_0}|\slashed\nabla^{k_1}\partial_v^{k_2}(r^2\partial_v)^{k_3}\Gamma\phi_{\ell = 1,2}|^2\Big)\, \mathrm{vol}_\Sigma \\
	&\qquad + \int\limits_{\Sigma} T(\psi)(n_\Sigma, n_\Sigma)(\Gamma \psi) \, \mathrm{vol}_\Sigma \;,
	\end{split}
	\end{equation}
	where $\phi = r\psi$, $\phi_{\ell = 1,2}$ denotes the projection of $\psi$ on the first and second spherical harmonics, $\partial_v$ is with respect to the standard Eddington-Finkelstein $(u,v, \theta, \varphi)$ coordinates in the exterior, and $\vol_\Sigma$ denotes the volume form of $\Sigma$. The sum is over a finite set of multi-indices $i,j,k$, and $\Gamma$ is a finite collection of products of the form $\partial_t^{(n)} \Omega_k^{(m)}$.
	Note that all the derivatives in \eqref{EnergyNorm} can be computed from $(\psi|_\Sigma, n_\Sigma \psi)$. As shown in \cite{AnArGa16b}, \cite{AnArGa18}, this class of initial data leads to decay of the arising solutions on the event horizon as in the assumptions of Theorem \ref{ThmMaster} with $q=3$. For the proof of Theorem \ref{ThmOpenSet} we will work with this class of initial data and show openness with respect to the topology induced by the energy norm \eqref{EnergyNorm}.\footnote{As mentioned in Theorem \ref{ThmAAG} a different class of initial data that allows for slower decay towards spacelike infinity $\iota^0$ only leads to decay along the horizon with $q=2$ in \eqref{ThmMasterCond}. Our proof below transfers directly to this class of initial data.}
	
	\begin{proof}[Proof of Theorem \ref{ThmOpenSet}:]
		Let $\Sigma_0$ be a spherically symmetric Cauchy hypersurface for the maximal analytic Schwarzschild spacetime as in Figure \ref{FigHypersurfaces} and let $\psi_0$ be a spherically symmetric solution to the wave equation arising from initial data in $\mathcal{E}$ such that
		\begin{equation*}
		\partial_t \psi_0 \gtrsim v^{-4}
		\end{equation*}
		holds on $\Hp \cap \{v \geq v_0\}$, for some $v_0 \geq 1$ -- and analogously, with the same positive sign, for the left event horizon. By \cite{AnArGa16b} such a solution exists. By Proposition \ref{PropRedShiftHorizon} there exists a $v_1 \geq v_0$  such that
		\begin{equation*}
		Y\psi_0 \gtrsim v^{-4}
		\end{equation*}
		holds on $\Hp \cap \{v \geq v_1\}$ -- and analogously for the left event horizon.
		\begin{figure}[h]
			\centering
			\def\svgwidth{7cm}
			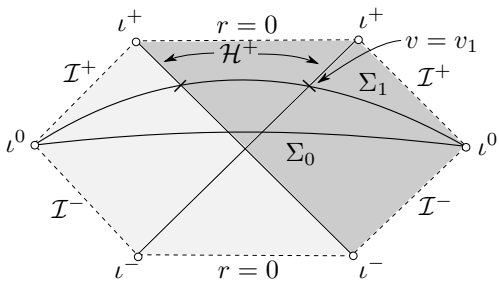
			\caption{The hypersurfaces $\Sigma_0$ and $\Sigma_1$ in the proof of Theorem \ref{ThmOpenSet}} \label{FigHypersurfaces}
		\end{figure}
		Let $\Sigma_1$ be a spherically symmetric Cauchy hypersurface for the maximal analytic Schwarzschild spacetime intersecting the right event horizon at $v = v_1$ -- and similarly for the left event horizon -- see Figure \ref{FigHypersurfaces}. Consider the induced initial data for $\psi_0$ on $\Sigma_1$. Note that we have $\partial_t \psi_0|_{\Hp}(v_1) >0$ and $Y\psi_0|_{\Hp}(v_1) >0$ -- and similarly for the left event horizon. We  can now change the spherically symmetric initial data 
		on the part of $\Sigma_1$ lying in the interior,
		such that $\partial_u \psi_0 >0$ and $\partial_v \psi_0 >0$ holds there. Here, we have denoted the new spherically symmetric solution again by $\psi_0$, and we continue doing so. By Theorem \ref{ThmSphSym} a) there exists now a $C_0>0$ such that
		\begin{equation*}
		\psi_0(t,r) \geq \psi_0(t,r_0) + C_0 (|t|+1)^{-4} \log \frac{r_0}{r}
		\end{equation*}
		holds in $\{0 < r < r_0\}$ for some $r_0 \in (0,2m)$.
		Moreover, by the results of \cite{AnArGa16b} the upper bounds in the assumptions of Theorem \ref{ThmExp} are satisfied with $q=3$ and $\delta = 1$ and thus $\psi_0$ has an expansion in $\{0 < r < r_0\}$ of the form 
		\begin{equation*}
		\psi_0(t,r) = A_0(t) \log r + B_0(t) + P_0(t,r)
		\end{equation*}
		with $|A_0(t)| \lesssim (|t| + 1)^{-4}$, $|B_0(t)| \lesssim (|t| + 1)^{-3}$, and $|P_0(t,r)| \lesssim r |\log r| (|t| +1)^{-3}$. It now follows as in the proof of Theorem \ref{ThmMaster} that the following holds
		\begin{equation}\label{C0}
		A_0(t) \leq -C_0 (|t| + 1)^{-4} \;.
		\end{equation}
		We now consider a ball $B_\lambda(\psi_0)$ of radius $\lambda$ in the energy \eqref{EnergyNorm} around the initial data induced by $\psi_0$ on $\Sigma_1$. It follows from the work \cite{AnArGa18} that for every solution $\psi$ of the wave equation arising from initial data  contained in $B_\lambda(\psi_0)$ the assumption \ref{AssumptionCorContDep} of Corollary \ref{CorContDepA} holds with $\delta =1$ and a $\mu$ depending on $\lambda$ that satisfies $\mu(\lambda) \to 0$ for $\lambda \to 0$. We can now choose $\lambda_0 >0$ small enough such that the constant $C\big(\mu(\lambda)\big)$ in \eqref{ConclusionCorContDep} satisfies $C\big(\mu(\lambda_0)\big) < \frac{C_0}{2}$. Together with \eqref{C0} it now follows that for every solution $\psi$ of the wave equation arising from initial data in $B_{\lambda_0}(\psi_0)$ on $\Sigma_1$ the expansion \eqref{AsymptoticExp} holds with $A (t, \omega) \leq -\frac{C_0}{2}(|t| + 1)^{-4}$. Clearly, this set of solutions also corresponds to an open set of initial data on $\Sigma_0$. This proves Theorem \ref{ThmOpenSet}.
	\end{proof}

	\section{Full asymptotic expansion of $\psi$ near $\{r=0\}$: proof of Theorem \ref{ThmFullExp}}\label{sec:fullexp}

	\begin{proof}[Proof of Theorem \ref{ThmFullExp}:]
		We shall prove the validity of the expansion (\ref{fullasymexp}) by induction. For $N=0$, the conclusion holds by Theorem \ref{ThmExp}, for $\zeta_0=A,\eta_0=B,R_0=P$. 
		Plugging (\ref{fullasymexp}) into the wave equation for $\psi$ we obtain:
		\begin{align}\label{boxvN}
		\square R_N
		+\frac{\log r}{r^2}\sum_{n=0}^{N-1}\bigg[n(n+1)\zeta_n
		-2m(n+1)^2\zeta_{n+1}+\Delta_{\mathbb{S}^2}\zeta_{n}+\sum_{l=0}^{n-3}(\frac{1}{2m})^l\frac{\partial_t^2\zeta_{n-3-l}}{2m}\bigg]r^n&\\ 
		\notag+\frac{1}{r^2}\sum_{n=0}^{N-1}\bigg[(2n+1)\zeta_{n}+n(n+1)\eta_{n}-2m(2n+2)\zeta_{n+1}
		-2m(n+1)^2\eta_{n+1}+\Delta_{\mathbb{S}^2}\eta_{n}\\
		\notag+\sum_{l=0}^{n-3}(\frac{1}{2m})^l\frac{\partial_t^2\eta_{n-3-l}}{2m}\bigg]r^n
		+\bigg[N(N+1)\zeta_{N}+\Delta_{\mathbb{S}^2}\zeta_{N}+\sum_{l=0}^{N-3}(\frac{1}{2m})^l\frac{\partial_t^2\zeta_{N-3-l}}{2m}\bigg]r^{N-2}\log r&\\
		\notag+\bigg[(2N+1)\zeta_N+N(N+1)\eta_N
		+\Delta_{\mathbb{S}^2}\eta_N+\sum_{l=0}^{N-3}(\frac{1}{2m})^l\frac{\partial_t^2\eta_{N-3-l}}{2m}\bigg]r^{N-2}+O(r^{N-1}|\log r|)=&\,0
		\end{align}
		Assuming $\zeta_{n+1},\eta_{n+1}$, $n=0,\ldots,N-1$ are given by (\ref{zetanetan}), the  equation \eqref{boxvN} becomes
		\begin{align}\label{boxvN2}
		\square R_N=&-\bigg[N(N+1)\zeta_{N}+\Delta_{\mathbb{S}^2}\zeta_{N}+\sum_{l=0}^{N-3}(\frac{1}{2m})^l\frac{\partial_t^2\zeta_{N-3-l}}{2m}\bigg]r^{N-2}\log r\\
		\notag&-\bigg[(2N+1)\zeta_N+N(N+1)\eta_N
		+\Delta_{\mathbb{S}^2}\eta_N+\sum_{l=0}^{N-3}(\frac{1}{2m})^l\frac{\partial_t^2\eta_{N-3-l}}{2m}\bigg]r^{N-2}+O(r^{N-1}|\log r|)
		\end{align}
		On the other hand, expressing the wave operator in coordinates we have 
		\begin{align}\label{boxvN3}
		\notag \square R_N=&\,(1-\frac{2m}{r})\partial^2_rR_N+\frac{2}{r}(1-\frac{m}{r})\partial_rR_N+\frac{1}{r^2}\Delta_{\mathbb{S}^2}R_N+(\frac{2m}{r}-1)^{-1}\partial^2_tR_N\\
		\notag =&\,\frac{1}{r}(1-\frac{2m}{r})\partial_r(r\partial_rR_N)+\frac{1}{r}\partial_rR_N+\frac{1}{r^2}\Delta_{\mathbb{S}^2}R_N+(\frac{2m}{r}-1)^{-1}\partial^2_tR_N\\
		=&-\bigg[N(N+1)\zeta_{N}+\Delta_{\mathbb{S}^2}\zeta_{N}+\sum_{l=0}^{N-3}(\frac{1}{2m})^l\frac{\partial_t^2\zeta_{N-3-l}}{2m}\bigg]r^{N-2}\log r\\
		\notag&-\bigg[(2N+1)\zeta_N+N(N+1)\eta_N
		+\Delta_{\mathbb{S}^2}\eta_N+\sum_{l=0}^{N-3}(\frac{1}{2m})^l\frac{\partial_t^2\eta_{N-3-l}}{2m}\bigg]r^{N-2}+O(r^{N-1}|\log r|)
		\end{align}
		Hence, solving for $\partial_r(r\partial_rR_N)$, integrating in $[0,r]$ and using our assumptions on $R_N$ it follows that
		\begin{align}\label{boxvN4}
		r\partial_rR_N=&-\int^r_0s(1-\frac{2m}{s})^{-1}\bigg[N(N+1)\zeta_{N}+\Delta_{\mathbb{S}^2}\zeta_{N}+\sum_{l=0}^{N-3}(\frac{1}{2m})^l\frac{\partial_t^2\zeta_{N-3-l}}{2m}\bigg]s^{N-2}\log sds+O(r^{N+1})\\
		\notag=&\,\frac{1}{2m(N+1)}\bigg[N(N+1)\zeta_{N}+\Delta_{\mathbb{S}^2}\zeta_{N}+\sum_{l=0}^{N-3}(\frac{1}{2m})^l\frac{\partial_t^2\zeta_{N-3-l}}{2m}\bigg]r^{N+1}\log r+O(r^{N+1})
		\end{align}
		Divide with $r$ and integrate in $[0,r]$ once more to deduce that
		\begin{align}\label{boxvN4}
		R_N=&\frac{1}{2m(N+1)^2}\bigg[N(N+1)\zeta_{N}+\Delta_{\mathbb{S}^2}\zeta_{N}+\sum_{l=0}^{N-3}(\frac{1}{2m})^l\frac{\partial_t^2\zeta_{N-3-l}}{2m}\bigg]r^{N+1}\log r+O(r^{N+1})
		\end{align}
		The last equation implies that $R_N r^{-N-1}(\log r)^{-1}$ has a limit as $r\rightarrow0$. Define
		\begin{align}\label{zetaN+1}
		\zeta_{N+1}:=\lim_{r\rightarrow0}\frac{R_N}{r^{N+1}\log r}=\frac{N(N+1)\zeta_N+\Delta_{\mathbb{S}^2}\zeta_N+\sum_{l=0}^{N-3}(\frac{1}{2m})^l\frac{\partial_t^2\zeta_{N-3-l}}{2m}}{2m(N+1)^2}
		\end{align}
		We proceed by setting 
		\begin{align}\label{wN}
		w_N=R_N-\zeta_{N+1}r^{N+1}\log r
		\end{align}
		and performing a similar procedure for $w_N$ to add another term in the expansion. Plugging (\ref{wN}) into (\ref{boxvN2}) and using (\ref{zetaN+1}) we arrive at the following equation for $w_N$:
		\begin{align}\label{boxwN}
		\square w_N=&-\bigg[(2N+1)\zeta_N+N(N+1)\eta_N-2m(2N+2)\zeta_{N+1}
		+\Delta_{\mathbb{S}^2}\eta_N+\sum_{l=0}^{N-3}(\frac{1}{2m})^l\frac{\partial_t^2\eta_{N-3-l}}{2m}\bigg]r^{N-2}\\
		\notag&+O(r^{N-1}|\log r|)
		\end{align}
		Expressing the LHS in coordinates as above and treating (\ref{boxwN}) as an ODE for $w_N$ in $r$ we derive:
		\begin{align}\label{boxwN2}
		r\partial_rw_N=&-\int^r_0s(1-\frac{2m}{s})^{-1}\bigg[(2N+1)\zeta_N+N(N+1)\eta_N-2m(2N+2)\zeta_{N+1}
		+\Delta_{\mathbb{S}^2}\eta_N\\
		\notag&+\sum_{l=0}^{N-3}(\frac{1}{2m})^l\frac{\partial_t^2\eta_{N-3-l}}{2m}\bigg]s^{N-2}ds+O(r^{N+2}|\log r|)\\
		\notag=&\,\frac{1}{2m(N+1)}\bigg[(2N+1)\zeta_N+N(N+1)\eta_N-2m(2N+2)\zeta_{N+1}
		+\Delta_{\mathbb{S}^2}\eta_N\\
		\notag&+\sum_{l=0}^{N-3}(\frac{1}{2m})^l\frac{\partial_t^2\eta_{N-3-l}}{2m}\bigg]r^{N+1}+O(r^{N+2}|\log r|)
		\end{align}
		which after dividing with $r$ and integrating in $[0,r]$ yields the formula:
		\begin{align}\label{boxwN3}
		w_N=&\,\frac{1}{2m(N+1)^2}\bigg[(2N+1)\zeta_N+N(N+1)\eta_N-2m(2N+2)\zeta_{N+1}
		+\Delta_{\mathbb{S}^2}\eta_N\\
		\notag&+\sum_{l=0}^{N-3}(\frac{1}{2m})^l\frac{\partial_t^2\eta_{N-3-l}}{2m}\bigg]r^{N+1}+O(r^{N+2}|\log r|)
		\end{align}
		This in turn implies that 
		\begin{align}\label{etaN+1}
		\eta_{N+1}:=&\,\lim_{r\rightarrow0}\frac{w_N}{r^{N+1}}\\
		\notag=&\,\frac{(2N+1)\zeta_N+N(N+1)\eta_N-2m(2N+2)\zeta_{N+1}
			+\Delta_{\mathbb{S}^2}\eta_N
			+\sum_{l=0}^{N-3}(\frac{1}{2m})^l\frac{\partial_t^2\eta_{N-3-l}}{2m}}{2m(N+1)^2},
		\end{align}
		which together with (\ref{zetaN+1}) confirms (\ref{zetanetan}) for $n=N$. Next, setting 
		\begin{align}\label{vN+1}
		R_{N+1}:=w_N-\eta_{N+1}r^{N+1}=R_N-\zeta_{N+1}r^{N+1}\log r-\eta_{N+1}r^{N+1}
		\end{align}
		and going back to (\ref{boxwN}) we find that $R_{N+1}$ satisfies the equation
		\begin{align}\label{boxvN+1}
		\square R_{N+1}=O(r^{N-1}|\log r|)
		\end{align}
		Furthermore, by our assumptions on $R_N$ and (\ref{vN+1}), $R_{N+1}$ satisfies the bounds $|\partial_t^{(i)}\Omega_k^{(j)}R_{N+1}|\lesssim r^{N+1}|\log r|$, for all $i,j$, and $|\partial_r\partial_t^{(i)}\Omega_k^{(j)}R_{N+1}|\lesssim r^N|\log r|$.
		
		We improve upon the later bounds by viewing (\ref{boxvN+1}) as an ODE for $R_{N+1}$ in $r$ and treating the spatial derivatives of $R_{N+1}$ as inhomogeneous terms:
		\begin{align}\label{ODEvN+1}
		\frac{1}{r}(1-\frac{2m}{r})\partial_r(r\partial_rR_{N+1})=&-\frac{1}{r}\partial_rR_{N+1}-\frac{1}{r^2}\Delta_{\mathbb{S}^2}R_{N+1}-(\frac{2m}{r}-1)^{-1}\partial_t^2R_{N+1}+O(r^{N-1}|\log r|)\\
		\notag=&\,O(r^{N-1}|\log r|)
		\end{align}
		Hence, multiplying (\ref{ODEvN+1}) with $r(1-\frac{2m}{r})^{-1}$ we obtain:
		\begin{align}\label{vN+1bounds}
		|\partial_rR_{N+1}|\lesssim r^{N+1}|\log r|,&&|R_{N+1}|\lesssim r^{N+2}|\log r.
		\end{align}
		The same bounds obviously hold after commuting with the Killing vector fields $\partial_t^{(i)}\Omega_k^{(j)}$. This confirms the inductive assumptions for the remainder $R_{N+1}$ in the expansion 
		\begin{align}\label{fullexpN+1}
		\psi=\sum_{n=0}^{N}\zeta_nr^n\log r
		+\sum_{n=0}^{N}\eta_nr^n+R_N\overset{(\ref{vN+1})}{=}\sum_{n=0}^{N+1}\zeta_nr^n\log r
		+\sum_{n=0}^{N+1}\eta_nr^n+R_{N+1}
		\end{align}
		with $\zeta_{N+1},\eta_{N+1}$ satisfying (\ref{zetaN+1}),(\ref{etaN+1}). Thus, the proof of Theorem \ref{ThmFullExp} is complete.
	\end{proof}
	\begin{remark}
		Instead of invoking Theorem \ref{ThmExp} in the proof of Theorem \ref{ThmFullExp} above to confirm the zeroth step in our induction argument, we could start from scratch and in fact use a similar ODE type of argument to derive anew the expansion (\ref{AsymptoticExp}). Indeed, using the logarithmic upper bounds for $\psi$ and its derivatives that we derived in Section \ref{subsubsec:loguppbd}, we could view the wave equation for $\psi$ as an ODE in $r$, treating the spatial derivatives of $\psi$ as inhomogeneous, lower order terms. This would enable us to define $A(t,\omega),B(t,\omega),P(t,r,\omega)$ and prove the relevant bounds for $P$.
	\end{remark}
	\section{The one to one correspondence $\mathcal{S}$ of solutions $\psi$ and expansion coefficients $A,B$: proof of Theorem \ref{ThmIso}}\label{isoS}

	So far we have established a detailed description of linear waves in a Schwarzschild interior from the point of view of the forward Cauchy problem and we provided a full asymptotic expansion (\ref{fullasymexp}) of solutions towards $r=0$, which is determined by its first two leading terms, i.e., $A,B$ via the recurrence relations (\ref{zetanetan}). 
	
	In this section we study the map $\mathcal{S}$, given by (\ref{S}), which takes a solution $\psi$ to $A,B$. It is immediate from Theorem \ref{ThmExp} that $\mathcal{S}$ is well defined. We first prove that $\mathcal{S}$ is surjective:
	
	Given any smooth functions $A,B$, decaying appropriately at $|t|=+\infty$, we show that there exists a corresponding solution $\psi$ to the wave equation in the interior, having the expansion (\ref{AsymptoticExp}) towards $r=0$.
	This is achieved via a backwards construction method.
	\begin{figure}[h]
		\centering
		\def\svgwidth{7cm}
		\includegraphics[scale=1.5]{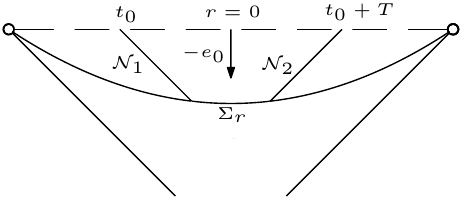}
		\caption{Backwards domain of dependence} \label{FigbackdomSch}
	\end{figure}

	Consider the following representation for $\psi$ which includes (\ref{AsymptoticExp}):\footnote{The necessity of inserting more terms in the expansion of $\psi$ is explained below in Remark \ref{rem:zetaeta}.}
	\begin{align}\label{Rdef}
	\psi(r,t,\omega)=&\,A(t,\omega)\log r+B(t,\omega)+\zeta_1(t,\omega)r\log r+\eta_1(t,\omega)r\\
	\notag&+\zeta_2(t,\omega)r^2\log r+\eta_2(t,\omega) r^2+R(r,t,\omega),
	\end{align}
	where we treat $R$ as an error term. The wave equation for $\psi$ then induces an inhomogeneous wave equation for $R$:  
	\begin{align}\label{boxR}
	\notag\square R+\frac{\log r}{r^2}(\Delta_{\mathbb{S}^2}A-2m\zeta_1 )
	+\frac{1}{r^2}(A-4m\zeta_1 +\Delta_{\mathbb{S}^2}B-2m\eta_1 )
	+\frac{\log r}{r}(2\zeta_1 +\Delta_{\mathbb{S}^2}\zeta_1-8m\zeta_2 )&\\
	+\frac{1}{r}(3\zeta_1 +2\eta_1 +\Delta_{\mathbb{S}^2}\eta_1-8m\zeta_2-8m\eta_2 )+\log r(\Delta_{\mathbb{S}^2}\zeta_2+6\zeta_2)+(5\zeta_2+6\eta_2+\Delta_{\mathbb{S}^2}\eta_2 )&\\
	\notag+(\frac{2m}{r}-1)^{-1}\partial_t^2(A\log r+B+\zeta_1  r\log r+\eta_1  r
	+\zeta_2r^2\log r+\eta_2 r^2)&=0
	\end{align}
	Our method of proof is based on deriving an a priori energy estimate for $R$ in suitable weighted $H^s$ spaces in the backward domain of dependence $D^{-}(\Sigma_{0,t_0,T})$ of $\Sigma_{0,t_0,T}:=\{r=0\}\cap\{t_0\leq t\leq t_0+T\}$, for any $t_0$ and fixed $T>0$, see Figure \ref{FigbackdomSch}. 
	\begin{theorem}\label{backthm}
		Let $A(t,\omega),B(t,\omega)\in C^\infty(\mathbb{R}\times\mathbb{S}^2)$ 
		such that
		\begin{align}\label{A,Bdec}
		|\partial_t^{(i)}\Omega_k^{(j)}A|\lesssim(|t|+1)^{-q-\delta},\qquad|\partial_t^{(i)}\Omega_k^{(j)}B|\lesssim(|t|+1)^{-q},
		\end{align}
		for $i+j\leq 8$. Then there exists a unique smooth function $\psi$ solving $\square_g\psi=0$ in the interior region $\{0<r\leq r_0\}$, $r_0<2m$, having the representation (\ref{Rdef}) with 
		\begin{align}\label{zetaeta}
		\begin{split}
		\zeta_1 :=\frac{\Delta_{\mathbb{S}^2}A}{2m},&\qquad\qquad\eta_1 :=\frac{A-4m\zeta_1 +\Delta_{\mathbb{S}^2}B}{2m},\\
		\zeta_2:=\frac{2\zeta_1+\Delta_{\mathbb{S}^2}\zeta_1}{8m},&\qquad\qquad\eta_2:=\frac{3\zeta_1+2\eta_1+\Delta_{\mathbb{S}^2}\eta_1-8m\zeta_2}{8m},
		\end{split}
		\end{align}
		$R\in C^\infty(M)$, satisfying $|R|\lesssim r^3|\log r|(|t|+1)^{-q}$ and
		\begin{align}\label{Rensp}
		r^{-\frac{5}{2}}\int_{\Sigma_r\cap D^{-}(\Sigma_{t_0,t_0+T})} J^{e_0}(\partial_t^{(i)}\Omega_k^{(j)}R)\mathrm{vol}_{\Sigma_r}
		\lesssim (1+T)(|t_0|+1)^{-2q}r^2|\log r|^2,
		\end{align}
		for all $r\in(0,2m)$, $t_0\in\mathbb{R}$, $T>0$, $i+j\leq 2$.
	\end{theorem}
	\begin{proof}
		For $\zeta_1 ,\eta_1 ,\zeta_2,\eta_2$ defined via (\ref{zetaeta}), the equation (\ref{boxR}) becomes
		\begin{align}\label{boxR2}
		\square R=&-\log r(\Delta_{\mathbb{S}^2}\zeta_2+6\zeta_2)-(5\zeta_2+6\eta_2+\Delta_{\mathbb{S}^2}\eta_2 )\\
		\notag&-(\frac{2m}{r}-1)^{-1}\partial_t^2(A\log r+B+\zeta_1  r\log r+\eta_1  r
		+\zeta_2r^2\log r+\eta_2 r^2)
		\end{align}
		We derive an a priori energy estimate for $R$ in $D^{-}(\Sigma_{0,t_0,T})\cap\{0<r\leq r_0\}$, $r_0< 2m$, with {\it trivial} initial data on $\Sigma_{0,t_0,T}$.
		
		Recall (\ref{pie0}) to compute:
		\begin{align}\label{Kr-5/2R}
		K^{r^{-\frac{5}{2}}e_0}( R)=&\,r^{-\frac{5}{2}}\pi^{e_0}_{\mu\nu}T^{\mu\nu}(\psi)-(e_0r^{-\frac{5}{2}})T_{00}(\psi)+r^{-\frac{5}{2}}e_0 R\cdot\square_g R\\
		\notag=&\,r^{-\frac{5}{2}}
		\frac{m}{r^2}(\frac{2m}{r}-1)^{-\frac{1}{2}}\frac{1}{2}\big[(e_0 R)^2+(e_1 R)^2-|\slashed{\nabla}  R|^2\big]
		-r^{-\frac{5}{2}}\frac{1}{r}(\frac{2m}{r}-1)^\frac{1}{2}\big[(e_0 R)^2-(e_1 R)^2\big]\\
		\notag&-(\frac{2m}{r}-1)^\frac{1}{2}\frac{5}{2}\frac{1}{r}r^{-\frac{5}{2}}\frac{1}{2}\big[(e_0 R)^2+|\overline{\nabla} R|^2\big]+r^{-\frac{5}{2}}e_0 R\cdot\square_g R
		\\
		\notag=&\,
		(\frac{2m}{r}-1)^\frac{1}{2}r^{-\frac{5}{2}} \frac{1}{r}\frac{1}{2}\bigg[\frac{m}{r}(\frac{2m}{r}-1)^{-1}\big[(e_0R)^2+(e_1R)^2-|\slashed{\nabla}R|^2\big]-2(e_0R)^2+2(e_1R)^2\\
		\notag&-\frac{5}{2}(e_0 R)^2-\frac{5}{2}|\overline{\nabla} R|^2\bigg]
		+r^{-\frac{5}{2}}e_0 R\cdot\square_gR\\
		\notag=&-(\frac{2m}{r}-1)^\frac{1}{2}r^{-\frac{5}{2}} \frac{1}{r}\big[2(e_0R)^2+\frac{3}{2}|\slashed{\nabla}R|^2\big]+O(1)(\frac{2m}{r}-1)^\frac{1}{2}r^{-\frac{5}{2}} J^{e_0}_0(R)+r^{-\frac{5}{2}}e_0 R\cdot\square_gR
		\end{align}
		{\it Initial data assumption for $R$}: 
		\begin{align}\label{Renlim}
		\lim_{r\rightarrow0}\int_{\Sigma_r\cap D^{-}(\Sigma_{0,t_0,T})}r^{-\frac{5}{2}}\big[(e_0\partial_t^{(i)}\Omega_k^{(j)}R)^2+|\overline{\nabla}\partial_t^{(i)}\Omega_k^{(j)}R|^2\big]\mathrm{vol}_{\Sigma_r}=0,
		\end{align}
		for $i+j\leq2$. Hence, applying the energy estimate with $X=r^{-\frac{5}{2}}e_0$ to $\partial_t^{(i)}\Omega_k^{(j)}R$, $i+j\leq2$, in the domain $D^{-}(\Sigma_{0,t_0,T})$, we arrive at the identity:
		\begin{align}\label{StokesR}
		&r^{-\frac{5}{2}} \int_{\Sigma_r\cap D^{-}(\Sigma_{0,t_0,T})}  J^{e_0}_0(\partial_t^{(i)}\Omega_k^{(j)}R)\mathrm{vol}_{\Sigma_r}+\sum_{l=1,2}\int_{\mathcal{N}_l}s^{-\frac{5}{2}}J^{e_0}_a(\partial_t^{(i)}\Omega_k^{(j)}R)\cdot (n_l)^a\mathrm{vol}_{\mathcal{N}_l}\\
		\notag=&\int^r_0(\frac{2m}{s}-1)^{-\frac{1}{2}}\int_{\Sigma_s\cap D^{-}(\Sigma_{0,t_0,T})}K^{s^{-\frac{5}{2}}e_0}(\partial_t^{(i)}\Omega_k^{(j)}R)\mathrm{vol}_{\Sigma_s}ds\\
		\notag=&\int^r_0\int_{\Sigma_s\cap D^{-}(\Sigma_{0,t_0,T})}
		-\frac{1}{s^\frac{7}{2}}\big[2(e_0\partial_t^{(i)}\Omega_k^{(j)}R)^2+\frac{3}{2}|\slashed{\nabla}\partial_t^{(i)}\Omega_k^{(j)}R|^2\big]
		+O(1) s^{-\frac{5}{2}} J^{e_0}_0(\partial_t^{(i)}\Omega_k^{(j)}R)\\
		\notag&+(\frac{2m}{s}-1)^{-\frac{1}{2}}s^{-\frac{5}{2}}e_0\partial_t^{(i)}\Omega_k^{(j)}R
		\bigg[-\log s(\Delta_{\mathbb{S}^2}\zeta_2+6\zeta_2)-(5\zeta_2+6\eta_2+\Delta_{\mathbb{S}^2}\eta_2 )\\
		\notag&-(\frac{2m}{s}-1)^{-1}\partial_t^2(A\log s+B+\zeta_1  s\log s+\eta_1 s
		+\zeta_2s^2\log s+\eta_2 s^2)\bigg]
		\mathrm{vol}_{\Sigma_s} ds
		\end{align}
		By the assumptions (\ref{A,Bdec}), (\ref{zetaeta}) we then have
		\begin{align}\label{Renineq}
		r^{-\frac{5}{2}} \int_{\Sigma_r\cap D^{-}(\Sigma_{0,t_0,T})}  J^{e_0}_0(\partial_t^{(i)}\Omega_k^{(j)}R)\mathrm{vol}_{\Sigma_r}
		\lesssim\int^r_0\int_{\Sigma_s\cap D^{-}(\Sigma_{0,t_0,T})}
		s^{-\frac{5}{2}}  J^{e_0}_0(\partial_t^{(i)}\Omega_k^{(j)}R)\mathrm{vol}_{\Sigma_s}ds&\\
		\notag+\int^r_0|\log s|(|t_0|+1)^{-q}\bigg(s^{-\frac{5}{2}}\int_{\Sigma_s\cap D^{-}(\Sigma_{0,t_0,T})}  (e_0\partial_t^{(i)}\Omega_k^{(j)}R)^2 \mathrm{vol}_{\Sigma_s}\bigg)^\frac{1}{2} ds&,
		\end{align}
		for $i+j\leq2$. Utilising now (\ref{Gron}), we arrive at (\ref{Rensp}). 
		
		By assumption $\partial_t^{(i)}\Omega_k^{(j)}R$ vanishes at $r=0$, $i+j\leq2$. Employing (\ref{Rensp}) we deduce the inequality:
		\begin{align*}
		\notag&\frac{1}{2}\partial_r\int_{\Sigma_r\cap D^{-}(\Sigma_{0,t_0,T})}(\partial_t^{(i)}\Omega_k^{(j)}R)^2dtd\mathbb{S}^2\\
		\leq& \bigg(\int_{\Sigma_r\cap D^{-}(\Sigma_{0,t_0,T})}(\partial_t^{(i)}\Omega_k^{(j)}R)^2dtd\mathbb{S}^2\bigg)^\frac{1}{2}\bigg(\int_{\Sigma_r\cap D^{-}(\Sigma_{0,t_0,T})}(\partial_r\partial_t^{(i)}\Omega_k^{(j)}R)^2dtd\mathbb{S}^2\bigg)^\frac{1}{2}\\
		\notag\lesssim&\sqrt{1+T}(|t_0|+1)^{-q}r^2|\log r|\bigg(\int_{\Sigma_r\cap D^{-}(\Sigma_{0,t_0,T})}(\partial_t^{(i)}\Omega_k^{(j)}R)^2dtd\mathbb{S}^2\bigg)^\frac{1}{2}
		\end{align*}
		or 
		\begin{align*}
		\partial_r\bigg(\int_{\Sigma_r\cap D^{-}(\Sigma_{0,t_0,T})}(\partial_t^{(i)}\Omega_k^{(j)}R)^2dtd\mathbb{S}^2\bigg)^\frac{1}{2}
		\lesssim \sqrt{1+T}(|t_0|+1)^{-q}r^2|\log r|
		\end{align*}
		Hence, integrating in $[0,r]$ we obtain the bound:
		\begin{align}\label{RL2est}
		\bigg(\int_{\Sigma_r\cap D^{-}(\Sigma_{0,t_0,T})}(\partial_t^{(i)}\Omega_k^{(j)}R)^2dtd\mathbb{S}^2\bigg)^\frac{1}{2}
		\lesssim \sqrt{1+T}(|t_0|+1)^{-q}r^3|\log r|,&&i+j\leq2,
		\end{align}
		from which it also follows that $|R|\lesssim\|R\|_{H^2}\lesssim r^3|\log r|(|t|+1)^{-q}$.
		
		Together with a standard duality type of argument that we omit, see \cite{AlGer}, the a priori energy estimate (\ref{Rensp}) yields a solution $R$ to (\ref{boxR}) in $D^{-}(\Sigma_{0,t_0,T})$, for any $t_0>0$ (Figure \ref{FigbackdomSch}). By the domain of dependence property, the solutions coincide in the overlapping regions and hence they produce a solution $\psi$ to the wave equation in the region $\{0<r\leq r_0\}$, for an arbitrary $r_0<2m$, having the representation (\ref{Rdef}) with $R$ satisfying (\ref{Rensp}). 
		This completes the proof of the theorem.
	\end{proof}
	\begin{remark}\label{rem:zetaeta}
		The asymptotic expansion (\ref{Rdef}) contains more terms than (\ref{AsymptoticExp}) which we obtained in the forward problem. The reason we had to include next order terms in the expansion of $\psi$ for the above construction is because of the singular inhomogeneous terms in the first line of (\ref{boxR}). Had we not included the functions $\zeta_1,\eta_1,\zeta_2,\eta_2$ in (\ref{Rdef}), defined via (\ref{zetaeta}), then the $|\log s|$ coefficient in the bulk integral in the last line of (\ref{Renineq}) would be of the order $|\log s|s^{-2}$ which is far from being integrable in $[0,r]$. This would not allow us to apply Gronwall's inequality and obtain an energy estimate.
	\end{remark}
	The following corollary is an immediate application of the previous theorem, along with Theorem \ref{ThmFullExp} for $N=2$.
	\begin{corollary}\label{thm:iso}
		Let $\psi_1,\psi_2$ be two smooth solutions of the wave equation in the interior region $\{0<r< 2m\}$, having the expansion:
		\begin{align}\label{exppsi1=exppsi2}
		\psi_l(r,t,\omega)=A(t,\omega)\log r+B(t,\omega)+P_l(r,t,\omega),&&l=1,2.
		\end{align}
		where $A,B$ are smooth functions and the error terms $P_l$, $l=1,2$, satisfy $|\partial_t^{(i)}\Omega_k^{(j)}P_l|\lesssim r|\log r|$, for all $i+j\leq 5$, $k=1,2,3$.
		Then $\psi_1=\psi_2$ in $\{0<r<2m\}$.
	\end{corollary}
	\begin{proof}[Proof of Theorem \ref{ThmIso}]
		From Corollary \ref{thm:iso} and Theorem \ref{backthm} it follows that $S$ is one to one and onto respectively. Thus, $S$ is an isomorphism between vector spaces. 
	\end{proof}

	\section*{Acknowledgements}
	The authors would like to thank Mihalis Dafermos for valuable discussions. They would also like to thank Yannis Angelopoulos, Stefanos Aretakis, Dejan Gajic for 
	helpful comments on their recent series of papers \cite{AnArGa16a,AnArGa16b,AnArGa18}. G.F. was supported by the EPSRC grant EP/K00865X/1 on `Singularities of Geometric Partial Differential Equations'. G.F. would also like to thank his former advisor, Spyros Alexakis, for numerous enlightening communications while this work was being derived. He would also like to thank Jonathan Luk and Volker Schlue for stimulating discussions. J.S. would like to thank Peter Hintz for a stimulating discussion.

\end{document}